\documentclass[11pt]{article}

\usepackage{amsmath}

\usepackage{times}
\usepackage{bm}
\usepackage{natbib}
\usepackage{amssymb}
\usepackage{graphicx}
\usepackage{accents}

\usepackage[margin=1.3in]{geometry}

\usepackage[plain,noend]{algorithm2e}
\usepackage{paralist}
\usepackage{amsthm}
\usepackage{authblk}
\usepackage{hyperref}
\hypersetup{
    colorlinks=true,
    linkcolor=blue,
    filecolor=blue,      
    urlcolor= blue,
    citecolor=blue,
}

\newtheorem{theorem}{Theorem}[section]

\newtheorem{proposition}{Proposition}[section]
\newtheorem{remark}[]{Remark}[section]

\newtheorem{corollary}{Corollary}[section]

\newlength{\dhatheight}
\newcommand{\doublehat}[1]{%
    \settoheight{\dhatheight}{\ensuremath{\hat{#1}}}%
    \addtolength{\dhatheight}{-0.25ex}%
    \hat{\vphantom{\rule{1pt}{\dhatheight}}%
    \smash{\hat{#1}}}}

\DeclareMathOperator*{\argmax}{arg\,max}

\newcommand{\field}[1]{\mathbb{#1}}

\begin{document}


\title{Feasible Invertibility Conditions for Maximum Likelihood Estimation for Observation-Driven Models\footnotemark \footnotetext{Corresponding author: Paolo Gorgi. 
Email address:  \texttt{gorgi@stat.unipd.it}}}

\author[a,b]{F. Blasques}
\author[a,c]{P. Gorgi}
\author[a,b,d]{S. J. Koopman}
\author[e,f]{O. Wintenberger}
\affil[a]{Vrije Universiteit Amsterdam, The Netherlands}
\affil[b]{Tinbergen Institute, The Netherlands}
\affil[c]{University of Padua, Italy}
\affil[d]{CREATES, Aarhus University, Denmark}
\affil[e]{Department of Mathematical Sciences, University of Copenhagen, Denmark}
\affil[f]{Sorbonne Universit\'{e}s, UPMC University Paris 06, France}

{\let\newpage\relax\maketitle}

\begin{abstract}
\noindent Invertibility conditions for observation-driven time series models often fail to be guaranteed in empirical applications. As a result, the asymptotic theory of maximum likelihood and quasi-maximum likelihood estimators may be compromised.  
We derive considerably weaker  conditions that can be used in practice to ensure the consistency of the maximum likelihood estimator for a wide class of observation-driven time series models. Our consistency results hold for both correctly specified and misspecified models. The practical relevance of the theory is highlighted in a set of empirical examples.  We further obtain an asymptotic test and  confidence bounds for the unfeasible ``true'' invertibility region of the parameter space.
\end{abstract}


\section{Introduction}

Observation-driven models are widely employed in time series analysis and econometrics. These models feature time-varying parameters that are specified through a stochastic recurrence equation (SRE) that is driven by past observations of the time series variable.  A more accurate description of this class of models is provided by \cite{cox1981}. A key illustration of the observation-driven model class is the Generalized Autoregressive Conditional Heteroscedasticity (GARCH) model as introduced by \cite{engle1982} and \cite{bol1986}. Observation-driven models are also widely employed  outside the context of volatility models; see, for instance, the dynamic conditional correlation (DCC) model of \cite{Engle2002}, the time-varying quantile model of \cite{EM2004}, the dynamic copula models of \cite{Patton2006}, the score-driven models of \cite{ckl2013} and the time-varying location model of \cite{HL2014}. 

The asymptotic theory of the Quasi Maximum Likelihood (QML) estimator for GARCH and related models has attracted much attention. \cite{Lum1996} and \cite{Lee94} obtained the consistency and asymptotic normality of the QML estimator for the GARCH(1,1). \cite{berkes2003} generalized their results to the GARCH$(p,q)$ model. Among others,  \cite{FZ2004} and \cite{robinson2006} weakened the conditions for consistency and asymptotic normality and extended the results to a larger class of models. \cite{SM2006}  have provided a general approach that allows to handle nonlinearities in the variance recursion. The theory relies on the work of \cite{Bougerol1993} to ensure the invertibility of the filtered time-varying variance and to deliver asymptotic results that are subject to some restrictions on the parameter region where the QML estimator is defined.  The severity of these restrictions typically depends on the degree of nonlinearity in the recurrence equation. 

The invertibility conditions of \cite{SM2006} often fail to be guaranteed in empirical studies. In Section \ref{sec2}  and  \ref{sec6} we illustrate this issue through some empirical examples featuring the Beta-$t$-GARCH$(1,1)$ model of \cite{H2013} and  \cite{ckl2013}, the dynamic autoregressive model of \cite{nlgas2014} and \cite{DMP204}, and the fat-tailed location model of \cite{HL2014}. The main problem is due to the conditions themselves since they depend on the unknown data generating process. Hence they cannot be verified in practice. This leads researchers to rely on  feasible conditions that are typically only satisfied in either degenerate or very small parameter regions, which are unreasonable in practical situations. To address this issue and to ensure the asymptotic theory of the QML estimator of the EGARCH(1,1) model of \cite{Nelson1991}, \cite{Win2013} proposed to stabilize the inferential procedure by restricting the optimization of the quasi-likelihood function to a parameter region that satisfies an empirical version of the required invertibility conditions of \cite{SM2006}. This method  provides a consistent QML estimator for the EGARCH(1,1) model.

In recent contributions, consistency proofs for observation-driven models with nonlinear filters have appeared that do not rely on the invertibility concept of \cite{SM2006}; see, for instance, \cite{H2013}, \cite{HL2014} and \cite{Ryoko2016}. However,  these results appeal to Lemma 2.1 of \cite{JRab2004} and rely on the restrictive and non-standard assumption that the  true value of the unobserved time-varying parameter is known at time $t=0$. Although \cite{JRab2004} carefully show that they do not need to impose this assumption in their results for the non-stationary GARCH model, this crucial issue is typically not addressed in other work. As it is discussed in \cite{Win2013} and \cite{sorokin2011}, invertibility is not just a technical assumption. The lack of knowledge of the time-varying parameter at $t=0$ can lead to the impossibility of recovering asymptotically the true time-varying parameter  even when the true static parameter vector is known. Furthermore, besides the invertibility issue, the results based on Lemma 2.1 of \cite{JRab2004} are only valid under the correct specification and by assuming that the likelihood function is maximized on an arbitrary small neighbourhood around the true parameter value. 

We extend the stabilization method of \cite{Win2013} to a large class of observation-driven models and prove the consistency of the resulting maximum likelihood (ML) estimator. These results hold for both correctly specified and incorrectly specified models, in the latter case a pseudo-true parameter is considered.  Additionally, we derive a test and confidence bounds for the ``true''  unfeasible parameter region. Our results cover a very wide class of models including  ML estimation of GARCH and related models. In financial applications, maximum likelihood estimation for the GARCH family of models is often preferred to QML estimation as  the time series exhibit fat-tails and asymmetry. In this context, we provide an example of how our results can be useful in practice. In particular, we prove the consistency of the ML estimator for the Beta-$t$-GARCH(1,1) model of \cite{H2013}. The usefulness of our theoretical results is further illustrated considering two examples in the context of dynamic location model. In particular, we discuss the implications of our theoretical results considering the dynamic autoregressive model of \cite{nlgas2014} and \cite{DMP204}  and the fat-tailed location model of \cite{HL2014}.

The paper is structured as follows. Section \ref{sec2} motivates the theory with an empirical application for which the invertibility conditions used in \cite{SM2006} are too restrictive.  Section \ref{sec3} introduces the notion of invertibility of the filter and analyzes it in the context of the class of observation-driven models. Section \ref{sec4} presents the asymptotic results. Section \ref{sec5} derives an invertibility test for the filter and obtains confidence bounds for the parameter space of interest. Section \ref{sec6} shows the practical importance of asymptotic results through some empirical illustrations.  Section \ref{sec8} concludes.

\section{Motivation}
\label{sec2}

Consider the  Beta-$t$-GARCH(1,1) model introduced by  \cite{H2013} and \cite{ckl2013}  for a sequence of financial returns $\{y_{t}\}_{t \in \field{N}}$ with time-varying conditional volatility and leverage effects,

\begin{equation}\label{tgass}
y_{t}=\sqrt{f_t}\mathcal{\varepsilon}_t , \qquad 
f_{t+1}= \omega+\beta f_t+(\alpha +\gamma d_t)\frac{(v+1)y_t^2}{(v-2)+ y_t^2/  f_t},
\end{equation}
  where $\{\varepsilon_t\}_{t\in \mathbb{Z}}$ is an i.i.d.~sequence of standard Student's $t$ random variables with $v>2$ degrees of freedom and $d_t$ is a dummy variable that takes value $d_{t}=1$ for $y_{t} \leq 0$ and $d_{t}=0$ otherwise. In order to perform ML estimation of the model, the observed data $\{y_t\}_{t=1}^n$ are used to obtain the filtered time-varying parameter $\hat f_t(\theta)$ as
\begin{eqnarray*}
\hat f_{t+1}( \theta)=\omega+\beta \hat f_t(\theta)+(\alpha +\gamma d_t)\frac{(v+1)y_t^2}{(v-2)+ y_t^2/\hat f_t(\theta)}, \;\; t\in \mathbb{N},
\end{eqnarray*}
where the recursion is initialized at $\hat f_{0}( \theta) \in [0,+\infty)$. The invertibility concept of \cite{SM2006} is concerned with  the stability of $\hat f_t(\theta)$, in particular, it ensures that asymptotically the filtered parameter $\hat f_t(\theta)$ does not depend on the initialization $\hat f_{0}( \theta)$.
Figure \ref{fig:invplot} illustrates the importance of the invertibility of the filter. The plots show differences between filtered volatility paths obtained from the  S\&P 500 returns for different initializations $\hat f_{0}(\theta)$.  The left panel shows a situation where  the filter is invertible and hence the effect of the initialization $\hat f_{0}(\theta)$ on $\hat f_{t}( \theta)$ vanishes as $t$ increases. The right panel shows that the effect of the initialization does not vanish when the filter that is not invertible.  
\begin{figure}[h!]
\center
\includegraphics[width=0.80\textwidth]{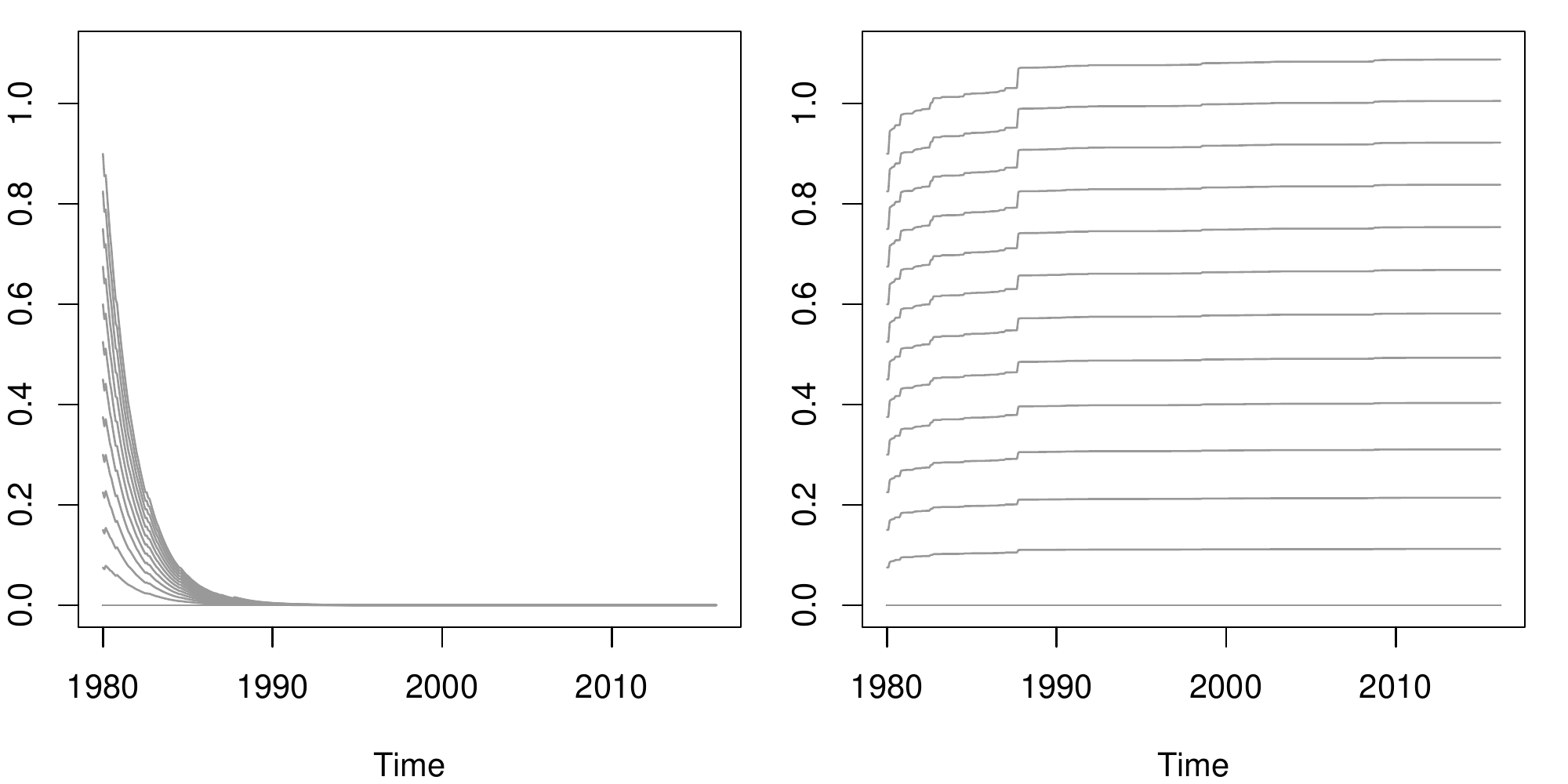}
\caption{\textit{The plots show differences of the filtered variance paths  for different initializations and using the  S\&P 500  time series. Differences are with respect to the filter initialized at $\hat f_0(\theta) = 0.1$. In the first plot, the vector of static parameters is selected to satisfy the invertibility conditions. In the second plot, a vector of static parameters that does not satisfy the invertibility conditions is considered.}}
\label{fig:invplot}
\end{figure}

From a ML estimation perspective, the lack of invertibility of the filter also poses fundamental problems. 
Without invertibility, even asymptotically, the likelihood function  depends on the initialization and hence this may lead the ML estimator to converge to different points when different initializations are considered. Furthermore, we may also be in a situation where we have a consistent estimator for the static parameter vector $\theta$ but not be able to consistently estimate the time-varying parameter. This consideration comes naturally from the fact that lack of invertibility can lead to the impossibility of recovering the true path of the time-varying parameter even when the true vector of static parameters $\theta_0$ is known, see \cite{Win2013} and \cite{sorokin2011} for a more detailed discussion.
As we shall see, the following condition is sufficient for invertibility, and hence ensures the reliability of the ML estimator,
\begin{eqnarray}\label{cons}
E\log\left|\beta+(\alpha+\gamma d_t)\frac{(v+1)y^4_t}{\left ((v-2)\bar \omega+y_t^2\right)^2}\right|<0, \; \forall \;\theta \in \Theta,
\end{eqnarray}
where $\bar\omega= \omega/(1-\beta)$. In practice, it is not possible to evaluate the expectation in (\ref{cons}) as it   depends on the unknown data generating process, even when the model is correctly specified since the true parameter vector $\theta_0$ is unknown. Therefore, the derivation of the region $\Theta$ has to rely on feasible sufficient conditions to ensure (\ref{cons}).
As we shall see in Section \ref{sec6}, assuming either correct specification or that $y_t$ has a symmetric probability distribution around zero\footnote{Without this assumption the feasible invertibility condition would be even more restrictive.}, we can obtain the following sufficient invertibility condition  that does not depend on $y_t$ 
\begin{eqnarray*}
\frac{1}{2}\log\left|\beta+(\alpha+\gamma)(v+1)\right|+\frac{1}{2}\log\left|\beta+\alpha(v+1)\right|<0.
\end{eqnarray*}
Figure \ref{fig:reg1} suggests that the set $\Theta$ obtained from such a sufficient  condition is too small for empirical applications. In particular, Figure \ref{fig:reg1} highlights that a typical ML  point estimate lies far outside $\Theta$. The specific point estimates are obtained from the Beta-$t$-GARCH model applied to a monthly time series of log-differences of the S\&P 500 financial index for a sample period from January 1980 to April 2016.  
\begin{figure}[h!]
\center
\includegraphics[width=0.85\textwidth]{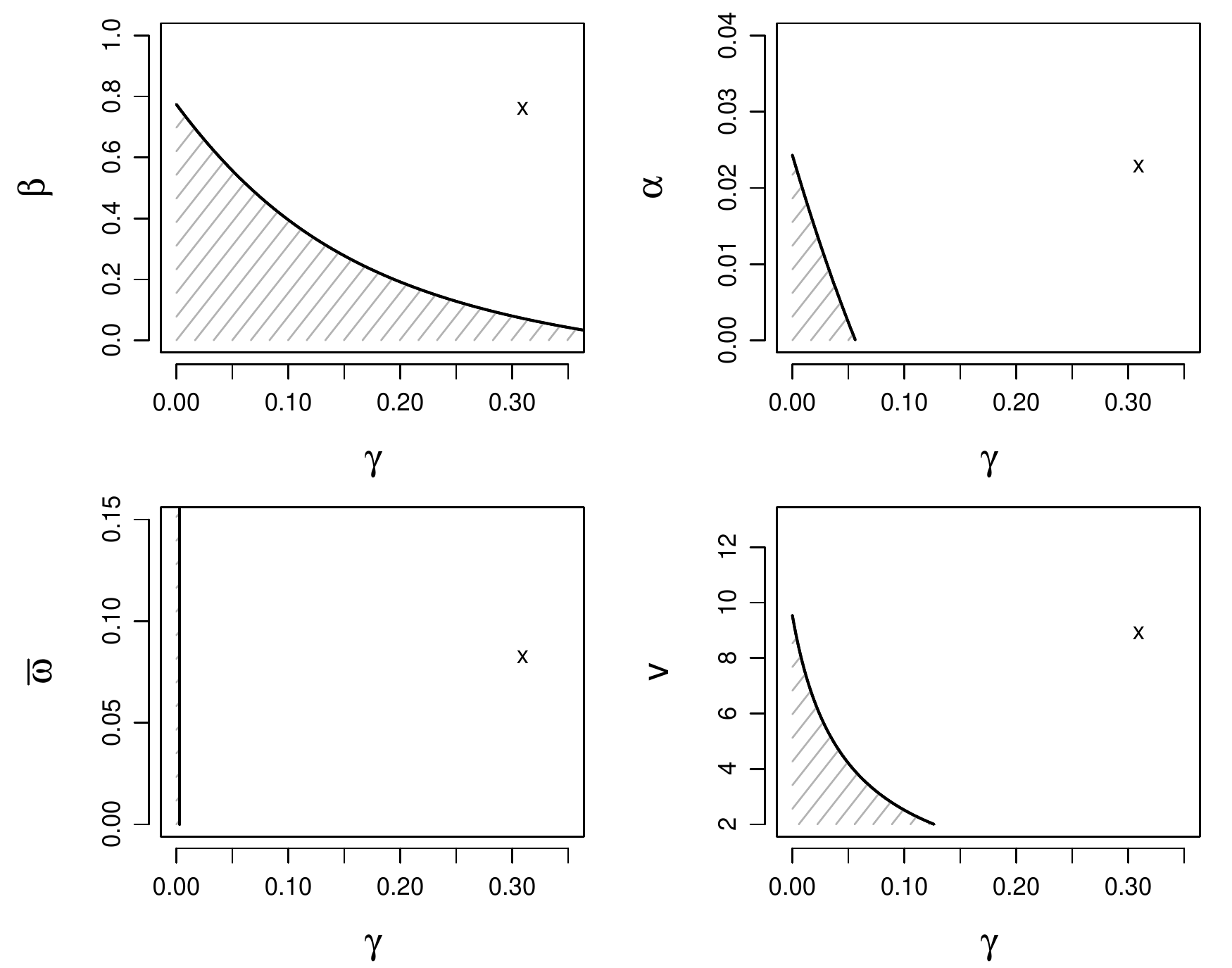}
\caption{\textit{ The shaded area identifies the parameter region $\Theta$ that satisfies sufficient conditions for invertibility. The crosses locate the point estimate of the parameters of the Beta-$t$-GARCH(1,1) model. }}
\label{fig:reg1}
\end{figure}
A visual inspection of Figure \ref{fig:reg1} may suggest that the presented point estimates reveal that the filter is not stable or invertible but in Section \ref{sec6} we will argue that this is not the case. These point estimates lie well inside the estimated regions for an invertible filter. in Section \ref{sec5} we develop the appropriate tests and confidence bounds which further confirm this claim.

The problem illustrated in Figure \ref{fig:reg1} is not specific to this sample of data or this conditional volatility model, see the discussion in Section \ref{sec6}. Different samples of financial returns produce similar point estimates that lie also outside $\Theta$. This problem is also not specific for the class of conditional heteroscedastic models. We illustrate this point   considering the autoregressive  model of \cite{nlgas2014} and \cite{DMP204} and the location model of   \cite{HL2014}.  We find that, in general, the typical invertibility conditions needed to ensure the consistency of the ML estimator, which are considered for  instance in \cite{SM2006}, \cite{Straumann2005} and \cite{BSA2014},  lead often to a parameter region that is too small for practical purposes. In contrary,   the estimation   method of \cite{Win2013}, proposed for the QML estimator  of the EGARCH(1,1) model, can provide a  parameter region large enough for practical applications.  In Section \ref{sec3} and Section \ref{sec4}, we  generalize the method of  \cite{Win2013} to  ML estimation of a wide class of observation driven models.


\section{Invertibility of observation-driven filters}
\label{sec3}

Let the observed sample of data $\{y_1,\dots ,y_n\}$ be a subset of the realized path of a random sequence $\{y_{t}\}_{t \in \field{Z}}$ with unknown conditional density  $p^{o}(y_{t}|y^{t-1})$, where $y^{t-1}$ denotes the entire past of the process $y^{t-1}:=\{y_{t-1},y_{t-2},...\}$. 
Consider the parametric observation-driven time-varying parameter model that is postulated by the researcher as given by
\begin{eqnarray}
&& y_{t}|f_t \sim p(y_t|f_t ,\theta), \label{model}\\
&& f_{t+1} =\phi(f_t ,Y^k_t,\theta),  \;\; t\in \mathbb{Z},\label{modelb}
\end{eqnarray}
where $\theta \in \Theta\subseteq \mathbb{R}^p$ is a vector of static parameters, $f_t$ is a time-varying parameter that takes values in $\mathcal{F}_\theta\subseteq \mathbb{R}$, $\phi$ is a continuous function from $\mathcal F_\theta \times \mathcal{Y}^k\times \Theta$ into $\mathcal F_\theta $, differentiable on its first coordinate, $Y_{t}^{k}$ is a vector containing at time $t$ the current and $k$ lags of the observed time series, that is $Y_{t}^{k}:=(y_{t},y_{t-1},...,y_{t-k})^T$, and  $p(\cdot|f_t,\theta)$ is a conditional density function such that $(y,f,\theta)\mapsto p(y|f,\theta)$ is continuous on $\mathcal{Y} \times \mathcal{F}_\theta \times  \Theta$. 

In general, we allow the parametric model  in (\ref{model}) and (\ref{modelb}) to be fully misspecified. It implies that both the dynamic specification of $f_t$ and the conditional density $p(\cdot|f_t,\theta)$ can be misspecified. A true time-varying parameter $f_t$ may not even exist because we only assume that a true conditional density $p^{o}(\cdot|y^{t-1})$ exists. When we assume correct specification,  the data generating process $\{y_{t}\}_{t \in \field{Z}}$ satisfies the model equations (\ref{model}) and (\ref{modelb}) for $\theta=\theta_0$ and we denote the true time-varying parameter as $f^o_t$. In this situation, we have that $p^o(\cdot|y^{t-1})=p(\cdot|f_t^o,\theta_0)$.

Despite the possibility of model misspecification, we emphasize that the model class based on (\ref{model}) and (\ref{modelb}) is general and covers a wide range of observation-driven models. It includes many GARCH and related models, the location models of \cite{HL2014}, the multiplicative error memory (MEM) model of \cite{Engle2002}, the autoregressive conditional duration model of \cite{EngleRussel1998}, the autoregressive conditional intensity model of \cite{Russell2001} and the Poisson autoregressive model of \cite{Davis2003}. 

An important advantage of observation-driven models is that the likelihood function is analytically tractable and it can be written in closed form as the product of conditional density functions. We consider the convention that the observations are available from time $t=1-k$.
Using the observed data, the filtered parameter $\hat{f}_{t}(\theta)$ that enters in the likelihood function  is obtained from the stochastic recurrence equation (SRE) given by
\begin{eqnarray}\label{rec}
\hat{f}_{t+1}(\theta)=\phi(  \hat{f}_{t}(\theta),  Y^k_t,  \theta ), \;\; t\in\mathbb{N},
\end{eqnarray}
where the recursion is initialized at $t=0$ with $\hat{f}_{0}({\theta}) \in \mathcal{F}_\theta$.
The set $\mathcal{F}_\theta$, where the time-varying parameter takes values, is indexed by $\theta \in \Theta$. As we will see for the Beta-$t$-GARCH model, this can be relevant in practice when dealing with specific models to weaken invertibility conditions; see the discussion in \cite{BSGW2015}. The ML estimator is then obtained as 
\begin{eqnarray}\label{mle}
\hat \theta_n(\hat f_0) = \argmax_{\theta \in \Theta}\hat{L}_n(\theta),
\end{eqnarray}
where $\hat{L}_n(\theta)$ denotes the log-likelihood function evaluated at $\theta \in \Theta$,
\begin{eqnarray}\label{llik}
\hat L_n( \theta) \ = \ n^{-1}\sum_{t = 1}^n \hat l_t(\theta) \ = \ n^{-1}\sum_{t = 1}^n \log p(y_t|\hat{f}_t(\theta),\theta) .
\end{eqnarray}

One of the difficulties in ensuring the consistency of the ML estimator is related to the recursive nature of the time-varying parameter and the consequent need of initializing the recursion in (\ref{rec}). In particular, the sequence $\{\hat f_t(\theta)\}_{t \in \mathbb{N}}$ as well as the sequence $\{\hat l_t(\theta)\}_{t\in \mathbb{N}}$ are both non-stationary. Therefore, the study of the limit behavior of $\{\hat f_t(\theta)\}_{t \in \mathbb{N}}$ is a natural requirement to ensure an appropriate form of convergence of the log-likelihood function $\hat L_n(\theta)$. 

\cite{Bougerol1993} provides well-known conditions for the filtered sequence $\{\hat{f}_{t}(\theta)\}_{t \in \field{N}}$ initialized at time $t=0$ to converge exponentially fast almost surely  (e.a.s.) to a unique stationary and ergodic sequence $\{\tilde f_{t}(\theta)\}_{t \in \field{Z}}$ as $t \to \infty$. In essence, this means that the effect of the initialization vanishes asymptotically at an exponential rate.\footnote{In the context of correctly specified models this implies that the true path $\{f_{t}^o\}_{t \in \field{Z}}$ can be asymptotically recovered as  $\hat{f}_{t}(\theta_{0})$ converges to $\tilde f_t(\theta_0)=f_{t}^o$ a.s.~as $t \to \infty$.} More formally, for any given $\theta \in \Theta$ and under appropriate conditions, Theorem 3.1 in \cite{Bougerol1993}  shows that  
\begin{equation*}
	|\hat f_t(\theta)-\tilde f_t(\theta)|\xrightarrow{e.a.s.} 0, \;\; t\xrightarrow{} \infty,
\end{equation*}
for any initialization $\hat f_0(\theta) \in \mathcal{F}_\theta$. \cite{SM2006} make use of Bougerol's theorem. Further, the e.a.s.~convergence stated above is sufficient for the invertibility of the filter\footnote{\cite{SM2006} say that the model is invertible if $\hat f_t(\theta_0)$ converges in probability to  $\tilde f_t^o$ and use Theorem 3.1 of \cite{Bougerol1993} precisely to obtain the desired convergence. }. Their definition of invertibility is  closely related to the definition of invertibility in \cite{Granger197887} since it implies  that $f_t^o$ is $y^{t-1}$ measurable. 

The stationary and ergodic limit sequence is denoted by $\tilde f_t(\theta)$ and it is \emph{not} denoted by $f_t(\theta)$ in order to stress that the stochastic properties of $\tilde f_t(\theta)$ are different from the stochastic properties of the sequence  $f_t(\theta)$ as implied by the model equations (\ref{model}) and (\ref{modelb}). This distinction is important as it emphasizes that $\tilde f_t(\theta)$ is driven by past random variables of the data generating process which are different than variables generated by the model equations (\ref{model}) and (\ref{modelb}). Under correct specification, we have that $\tilde f_t(\theta)$ has the same stochastic properties of $f_t(\theta)$ only when $\theta=\theta_0$ as the data generating process follows the model equations only at $\theta_0$. For more details, we refer to the discussions in \cite{SM2006} and \cite{Win2013}. 

Different conditions are required to establish invertibility and stationarity, even when the model is assumed to be well specified. As shown by \cite{sorokin2011} for models in the GARCH family, the situation can arise that, for a given $\theta_0$ value, the model in (\ref{modelb}) admits a stationary solution but it lacks an invertibility solution. In such a situation, the true sequence $\{\hat f_t(\theta_0)\}_{t \in \mathbb{N}}$ can exhibit chaotic behaviour and  the true path of $f_t^o$ cannot be recovered asymptotically even when the true vector of static parameters $\theta_0$ is known; see also the discussion in \cite{Win2013}. For this reason, ensuring the invertibility of the filtered parameter is not merely a technical requirement but an important ingredient to establish the reliability of the inferential procedure.


The invertibility of the the sequence $\{\hat f_t(\theta)\}_{t\in \mathbb{N}}$ evaluated at a single  parameter value $\theta \in \Theta$ is not enough to ensure an appropriate convergence of the log-likelihood function over $\Theta$. This happens naturally because the log-likelihood function depends on the functional sequence  $\{\hat f_t\}_{t\in \mathbb{N}}$. 
In this regard, \cite{Win2013} introduces the notion of continuous invertibility for GARCH-type models to ensure the uniform convergence of the filtered volatility. Accounting for the continuity of the function $\phi$,  the  elements of $\{\hat f_t\}_{t\in \mathbb{N}}$  can be considered as random elements in the space of continuous functions $\field{C}(\Theta,\mathcal{F}_\Theta)$, $\mathcal{F}_\Theta:=\bigcup_{\theta\in\Theta}\mathcal{F}_\theta$,  equipped with the uniform norm $\| \cdot \|_{\Theta}$,   $\| f \|_{\Theta} = \sup_{\theta \in \Theta}|f(\theta)|$ for any $f \in \field{C}(\Theta,\mathcal{F}_\Theta)$. Then the filter $\{\hat f_t\}_{t\in \mathbb{N}}$ is {\bf continuously invertible} if for any initialization $\hat  f_0 \in \field{C}(\Theta,\mathcal{F}_\Theta)$ we have
$$\|\hat f_t- \tilde f_t\|_\Theta \xrightarrow{e.a.s.} 0, \;\; t\xrightarrow{} \infty,$$
where $\{\tilde f_t\}_{t \in \mathbb{Z}}$ is a stationary and ergodic sequence of random functions.
This definition is related with the invertibility concept in \cite{Granger197887} as the invertibility implies that the stochastic function $\tilde f_t$ is $y^{t-1}$ measurable.  

Proposition \ref{pp1} presents sufficient conditions for the invertibility of $\{\hat f_t\}_{t\in \mathbb{N}}$. As in  \cite{Straumann2005}, \cite{SM2006} and \cite{Win2013},  the  conditions we consider  are based on Theorem 3.1 of \cite{Bougerol1993}.  First, we define  the  stochastic Lipschitz  coefficient $\Lambda_t(\theta)$ as
$$\Lambda_t(\theta) := \sup_{f \in \mathcal{F}_\theta}{\left|\dot \phi(f, Y_t^k, \theta)\right|},$$
where $\dot\phi(f, Y_t^k,\theta)=\partial \phi(f, Y_t^k,\theta)/\partial f$. 

\begin{proposition}\label{pp1}
Assume $\{y_t\}_{t \in \mathbb{Z}}$ is a stationary and ergodic sequence of random variables. Moreover, let the following conditions hold
\begin{enumerate}
  \item[(i)] There exists $\bar f\in \mathcal{F}_\Theta$ such that $E\log^+\|\phi(\bar f , Y_t^k, \cdot)\|_\Theta<\infty$.
  \item[(ii)] $E\sup_{\theta \in \Theta}\sup_{f \in \mathcal{F}_\Theta} \log^+\big|\dot\phi(f, Y_t^k,\theta) \big|<\infty$.
  \item[(iii)] $\log \Lambda_0(\theta)$ is a.s.~continuous on $\Theta$  and $E\log \Lambda_0(\theta)<0$ for any $\theta \in \Theta$.
  \end{enumerate}
   Then, the filter $\{\hat f_t\}_{t \in \mathbb{N}}$ is continuously invertible.
\end{proposition}
Proposition  \ref{pp1} not only ensures the convergence of $\{\hat f_t\}_{t \in \field{N}}$ to a stationary and ergodic sequence $\{\tilde f_t\}_{t\in \mathbb{Z}}$ but also that this sequence is unique and therefore the initialization $\hat f_0$ is irrelevant asymptotically. We emphasize that Proposition \ref{pp1} holds irrespective of the correct specification of the model as it only requires that the data are generated by a stationary and ergodic process.  In most practical situations, the so-called `contraction condition' stated in (iii) is  the most restrictive condition and it  also imposes the most severe constraints on the parameter space  $\Theta$.   
\begin{remark}\label{pp2}
When the model is correctly specified and the filter continuously invertible, then the filter evaluated at $\theta_0$ converges to the true unobserved time-varying parameter  $\{f_{t}^{o}\}_{t \in \field{Z}}$, i.e.
$$|\hat f_t(\theta_0) -f_t^o| \xrightarrow{\text{e.a.s.}} 0\;\; \text{as} \;\; t\rightarrow \infty$$ for any initialization $ \hat{f}_0(\theta_0) \in  \mathcal{F}_{\theta_0}$.
\end{remark}
Remark \ref{pp2}  highlights an important  implication of Proposition \ref{pp1} under correct specification. We obtain that, knowing the vector of static parameters $\theta_0$, the true path of $f^o_t$ can be recovered asymptotically. 
The next result shows that it is sufficient to have an approximate sequence $\{\hat\theta_n\}$ of the true parameter:
\begin{proposition}\label{pp3}
When the model is correctly specified and Conditions (i), (ii) and (iii) of Proposition \ref{pp1} hold, if $E[\log^+\|\tilde f_0\|_\Theta]<\infty$ and $\hat\theta_n \xrightarrow{\text{a.s.}} \theta_0$ then 
$$|\hat f_t(\hat\theta_n) -f_t^o| \xrightarrow{\text{e.a.s.}} 0\;\; \text{as} \;\; n\ge t\rightarrow \infty$$ for any initialization $ \hat{f}_0(\hat\theta_0) \in  \mathcal{F}_{\hat\theta_0}$.
\end{proposition}
\begin{remark}\label{pp4}
It can be surprisingly difficult to check the sufficient condition of existence of logarithmic moments. An alternative sufficient set of conditions is provided by Theorem 7 of \cite{Win2013}: $\{y_t\}$ is geometrically $\alpha$-mixing and for some $r>2$  $$E\sup_{\theta \in \Theta}\sup_{f \in \mathcal{F}_\Theta} (\log^+\big|\dot\phi(f, Y_t^k,\theta) \big|)^r<\infty.$$
\end{remark}

\section{Maximum likelihood estimation}
\label{sec4}
The invertibility of the filter can be used to establish the consistency of the ML estimator defined in (\ref{mle}) over the parameter space $\Theta$. Furthermore, we also show that the consistency results still hold after replacing the set $\Theta$ with an estimated set $\hat{\Theta}_{n}$ that ensures an empirical version of the contraction condition $E\log\Lambda_0(\theta)<0$. We consider  both the case of correct specification and misspecification of the observation-driven model. Finally, we derive confidence bounds for the unfeasible set of $\theta$s that satisfy the contraction condition $E\log\Lambda_0(\theta)<0$.

The subsequent results are subject to the stationarity and ergodicity of the data generating process. In the case of correct specification, stationarity and ergodicity can be checked studying the properties of the data generating process, see  \cite{blasques2014se} for sufficient conditions for a wide class of observation driven processes. In the case of misspecification, we allow the data generating process to be any stationary and ergodic process; this comes instead of imposing data to be generated by a specific stationary and ergodic process.

\subsection{Consistency of the ML estimator}

The first consistency result we obtain is under the assumption of correct specification. 
  We denote the log-likelihood function evaluated at the stationary filtered parameter $\tilde f_t$ as $L_n(\theta)=n^{-1}\sum_{t=1}^n l_t(\theta)$, where $l_t(\theta)=\log p(y_t|\tilde f_t(\theta),\theta)$ and  we denote by $L$ the function  $L(\theta)=E\, l_0(\theta)$. The following conditions are considered.
\begin{description}
  \item[C1:] The data generating process, which satisfies the equations (\ref{model}) and (\ref{modelb}) with $\theta=\theta_0\in \Theta$,  admits a  stationary and ergodic solution and $E |l_0(\theta_0)|<\infty$.
  \item[C2:] For any $\theta\in \Theta$, $l_0(\theta_0)=l_0(\theta)$ a.s. if and only if $\theta=\theta_0$. 
  \item[C3:] Conditions (i)-(iii) of Proposition \ref{pp1} are satisfied for the compact set $\Theta \subset \mathbb{R}^p$.
  \item[C4:] There exists a stationary  sequence of random variables $\{\eta_t\}_{t\in \mathbb{Z}}$ with $E\log^+|\eta_0|<\infty$ such that  almost surely $ \|\hat l_t-l_t\|_\Theta \le \eta_t \|\hat f_t-\tilde f_t\|_\Theta$ for any $t \ge N$,  $N\in \mathbb{N}$.
 \item[C5:] $E\|l_0\vee 0\|_\Theta<\infty$.
\end{description}
Condition \textbf{C1} ensures that the data are generated by a stationary and ergodic process and imposes an integrability condition on predictive log-likelihood, which is needed to apply an ergodic theorem. Condition \textbf{C2} is a standard identifiability condition. Conditions \textbf{C3} and \textbf{C4} ensure  the a.s.~uniform convergence of $\hat L_n$ to $L_n$. Finally, Condition \textbf{C5} ensures that $L_n$ converges  to an upper semicontinuous function $L$. As also considered in \cite{SM2006}, this final argument replaces the well known uniform convergence argument, namely, the uniform convergence of $L_n$ to $L$. Condition \textbf{C5} is weaker than the conditions that are typically needed for uniform convergence and in many cases it holds automatically as $l_0(\theta)$ is bounded from above with probability 1. Theorem \ref{th1} guarantees  the strong  consistency of the ML estimator.

\begin{theorem}\label{th1}
Let the conditions \textbf{C1}-\textbf{C5} hold, then the maximum likelihood estimator defined in (\ref{mle})  is strongly consistent, i.e.  $$\hat{\theta}_n(\hat f_0)  \xrightarrow{\text{a.s.}}\theta_0,\;\; \;\;\;\; n\xrightarrow{}\infty,$$ for any initialization $ \hat{f}_0 \in \mathbb{C}(\Theta, \mathcal{F}_\Theta)$.
\end{theorem}
The proof is presented in the Appendix. In Section \ref{sec6},  the strong consistency of the Beta-$t$-GARCH model is simply proved by checking these conditions. 

Often, the main objective of time series modeling is to describe the dynamic behaviour of the observed data and predict future observations. For this purpose, it is of interest to study  the   consistency of the estimation of the time-varying parameter $f_t^o$  and the conditional density function $p(y|f_t^o,\theta_0)$, $y\in \mathcal{Y}$. This further highlights the importance of the invertibility of the filter as without invertibility it may be possible to estimate consistently the static parameters, as shown by \cite{JRab2004} for the non-stationary  GARCH(1,1), but it is not possible to estimate consistently the time-varying parameter and the conditional density function.  We consider  plug-in estimates for  the time-varying parameter,   given by $\hat f_t(\hat \theta_n(\hat f_0) )$, and for the conditional density function,  given by $p(y|\hat f_t(\hat \theta_n(\hat f_0) ), \hat \theta_n(\hat f_0) )$,  $y \in \mathcal{Y}$. The next  result shows the consistency of these plug-in estimators which is due to an application of Proposition \ref{pp3} and a continuity argument:
\begin{corollary}\label{co11}
Let the conditions \textbf{C1}-\textbf{C5} and $E[\log^+\|\tilde f_0\|_\Theta]<\infty$ be valid, then the plug-in estimator  $\hat f_t(\hat \theta_n(\hat f_0) )$ is strongly consistent, i.e.
$$|\hat f_t(\hat \theta_n(\hat f_0) )-f_t^o|\xrightarrow{a.s.} 0,\;\;\;\;\;  n\ge  t\rightarrow\infty.$$
Moreover, assume that $f \mapsto p(y|f,\theta)$ is uniformly continuous in $f_\Theta$, then the plug-in density estimator $p(y|\hat f_t(\hat \theta_n(\hat f_0) ), \hat \theta_n(\hat f_0) )$ is strongly consistent, i.e.
$$\big|p(y|\hat f_t(\hat \theta_n(\hat f_0) ), \hat \theta_n(\hat f_0) )-p(y| f_t^o,  \theta_0)\big|\xrightarrow{a.s} 0,\;\;\;\;\;  n\ge t\rightarrow\infty,$$
for any $y \in \mathcal{Y}$ and any initialization $ \hat{f}_0 \in \mathbb{C}(\Theta, \mathcal{F}_\Theta)$.
\end{corollary}
Corollary \ref{co11} shows that the time-varying parameter $f_t^o$ and the conditional density function $p(y|f^o_t,\theta_0)$, $y\in\mathcal{Y}$, can be consistently estimated. The extra  logarithmic moments condition can be replaced by the set of conditions described in Remark \ref{pp4}.

\subsection{ML on an estimated parameter region} 
We have discussed it before, the Lyapunov condition  $E\log \Lambda_0(\theta)<0$ imposes some restriction on the parameter region $\Theta$ and, in situations where  $\Lambda_0(\theta)$ depends on $Y_0^k$, it cannot be checked as the expectation depends on the unknown data generating process. This also applies to the case of correct specification as the true parameter $\theta_0$ is unknown. A possible solution is to obtain testable sufficient conditions such that $E\log \Lambda_0(\theta)<0$ and to define the set  $\Theta$ accordingly. However,  this often leads to very severe restrictions, reducing the set $\Theta$ to a small region, which is too small for practical applications. An alternative is to check the condition $E\log \Lambda_0(\theta)<0$ empirically and to define the ML estimator  as the maximizer of the log-likelihood on an estimated parameter region. In the context of QML estimation, this approach have been proposed by \cite{Win2013} to stabilize the QML estimator of the EGARCH$(1,1)$ model of \cite{Nelson1991}. Here we formally define this maximum likelihood estimator and we prove its consistency for the general class of observation driven models defined in (\ref{model}). In Section \ref{sec6}, we show how these results can be relevant in practical applications.

We define a compact set $\hat \Theta_n$ that satisfies an empirical version of  the Lyapunov condition $E\log \Lambda_0(\theta)<0$,
\begin{eqnarray}\label{eps}
\hat \Theta_{n} = \left\{\theta \in  \bar \Theta: \frac{1}{n}\sum_{t=1}^n\log\Lambda_t(\theta)\le -\delta\right\},
\end{eqnarray}
where $\bar \Theta \in \mathbb{R}^p$ is a compact set and $\delta>0$ is an arbitrary small constant. We consider that the compact set $\bar \Theta$  is chosen in such a way that  $(f,y,\theta)\mapsto \phi (f,y,\theta)$ is continuous on $\mathcal F_{\bar{\Theta}} \times \mathcal{Y}^k\times {\bar{\Theta}}$ and  $(y,f,\theta)\mapsto p(y|f,\theta)$ is continuous on $\mathcal{Y} \times \mathcal{F}_{\bar{\Theta}} \times  \bar\Theta$.  For notational convenience, we also define the set $\Theta_c=\{\theta \in \bar\Theta: E\log \Lambda_0(\theta)<-c \}$, $c\in \mathbb{R}$.
The ML estimator on this empirical region $\hat \Theta_n$ is formally defined as
\begin{eqnarray}\label{smle}
{\doublehat \theta}_n(\hat f_0) =\argmax_{\theta \in \hat \Theta_n}\hat{L}_n(\theta).
\end{eqnarray}
To ensure the consistency of this ML estimator in the case of correct specification, the following conditions are considered.
\begin{description}
 \item[\textbf{A1:}] The data generating process, which is given by the model $(\ref{model})$ with $\theta_0\in \Theta_\delta$,  admits a  stationary and ergodic solution and $E |l_0(\theta_0)|<\infty$.
  \item[\textbf{A2:}] Condition \emph{(i)} and \emph{(ii)} of Proposition \ref{pp1} are satisfied for any compact subset $\Theta\subseteq \Theta_0$. Moreover, the map $\theta\mapsto \log\Lambda_0(\theta)$ is almost surely continuous on $\bar \Theta$ and $E\|\log\Lambda_0\|_{\bar \Theta}<\infty$.
  \item[\textbf{A3:}]  Conditions \textbf{C2}, \textbf{C4} and \textbf{C5} are satisfied for any compact subset $\Theta\subseteq \Theta_0$.
\end{description}
Condition \textbf{A1} ensures  that  stationarity,  ergodicity  and  invertibility of the data generating process. This condition can be seen as  the equivalent of the condition \textbf{C1} in Theorem \ref{th1} The condition \textbf{A2} imposes some assumptions on $\log\Lambda_0(\theta)$. These assumptions are needed to guarantee a certain form of convergence for the set $\hat\Theta_n$ and consequently ensure the continuous invertibility $\|\hat f_t-\tilde f_t\|_{\hat \Theta_n}\xrightarrow{\text{e.a.s.}}0$ as $t \rightarrow 0$ for large enough $n$. Therefore, \textbf{A2} can be seen as the equivalent of \textbf{C3} in Theorem 4.1. Finally,  \textbf{A3}, together with \textbf{A2}, is sufficient to ensure that asymptotically the identifiability condition \textbf{C2}, the regularity condition \textbf{C4} and the integrability condition \textbf{C5} hold. The next theorem states the strong consistency of the ML estimator in (\ref{smle}) under correct specification. 
\begin{theorem}\label{th2}
Let conditions \textbf{A1}-\textbf{A3} hold, then the maximum likelihood estimator defined in (\ref{smle}) is strongly consistent, i.e.  $$\doublehat{ \theta}_n(\hat f_0) \xrightarrow{\text{a.s.}}\theta_0,\;\; \;\;\;\; n\xrightarrow{}\infty$$ for any initialization  $ \hat{f}_0 \in \mathbb{C}(\bar \Theta, \mathcal{F}_{\bar \Theta})$.
\end{theorem}
Theorem \ref{th2} generalizes Theorem 5 of \cite{Win2013}, which is specific to QML estimation of the EGARCH(1,1) model, to ML estimation of the wide class of observation-driven models specified in (\ref{model}) and (\ref{modelb}). The conditions required to ensure the strong consistency in Theorem \ref{th2} are feasible to  be checked in practice. This differs from other results in the literature such as \cite{SM2006}, \cite{H2013}, \cite{HL2014} and \cite{Ryoko2016}.

We now switch our focus to the possibility of having a misspecified model. This case  is probably the most interesting one from a practical point of view as the assumption that the observed data are actually generated by the postulated model may be unreasonable. In the following, we show that, under misspecification, the ML estimator in (\ref{smle}) converges to a pseudo-true parameter $\theta^{*}$ that minimizes an average  Kullback-Leibler (KL) divergence  between the true conditional density $p^o(y_t|y^{t-1})$ and the postulated conditional density $p(y_t|\tilde f_t(\theta),\theta)$. Studies on consistency results with  respect to the pseudo true parameter for misspecified models go back to \cite{white1982}. We define the conditional KL divergence $KL_t(\theta)$  as
\begin{eqnarray}\label{KL}
KL_t(\theta) = \int_\mathcal{Y} \log \frac{p^o(x|y^{t-1})}{ p(x|\tilde f_t(\theta),\theta)}p^o(x|y^{t-1})dx
\end{eqnarray}
and the average (marginal) KL divergence $KL(\theta)$ as $KL(\theta)=E\, KL_t(\theta)$. The pseudo true parameter $\theta^*$ is defined as the minimizer of $KL(\theta)$.
The consistency result in this  misspecified framework follows the case of correct specification in a similar way because Proposition \ref{pp1} ensures the uniform convergence of $\hat f_t$ with no regards of the correct specification. The differences concern the stationarity and ergodicity  of the data generating process and the identifiability of the model. The following conditions are considered.
\begin{description}
  \item[M1:] The observed data are generated by a stationary and ergodic process $\{y_t\}_{t\in\mathbb{Z}}$ with conditional density function $p^o(y_t|y^{t-1})$ and the  condition  $E|\log p^o(y_0|y^{-1})|<\infty$ is satisfied.
  \item[M2:]  There is a parameter vector $\theta^*\in \Theta_\delta$   that is the unique maximizer of $L$, i.e.~$L(\theta^*)>L(\theta)$ for any $\theta\in \Theta_0$, $\theta\neq\theta^*$.
\item[M3:]  Condition $\textbf{A2}$ is satisfied and  $\textbf{C4}$ and $\textbf{C5}$  are satisfied for any compact set $\Theta\subseteq \Theta_0$.
\end{description}

Condition \textbf{M1} imposes the stationarity and ergodicity of the generating process and some moment conditions. Condition \textbf{M2}  ensures identifiability in this misspecified setting. The continuous invertibility is  ensured by \textbf{M3} as it imposes that \textbf{A2} holds while the results of Proposition \ref{pp1} are irrespective of the correct specification of the model. Finally, in the same way as in \textbf{A3}, \textbf{M3}  ensures that the conditions \textbf{C4} and \textbf{C5} hold for large enough $n$. 
\begin{theorem}\label{th22}
Let the conditions \textbf{M1}-\textbf{M3} hold, then the average KL divergence $KL(\theta)$ is well defined and the pseudo true parameter $\theta^*$ is its unique minimizer. Furthermore, the maximum likelihood  estimator  defined in (\ref{smle}) is strongly consistent, i.e.  $$ {\doublehat \theta}_n(\hat f_0) \xrightarrow{\text{a.s.}}\theta^*,\;\; \;\;\;\; n\xrightarrow{}\infty$$ for any initialization $ \hat{f}_0 \in \mathbb{C}(\bar \Theta, \mathcal{F}_{\bar \Theta})$.
\end{theorem}
This result further highlights the relevance of ensuring invertibility. In this case, it is not possible to assume correct initialization of the filtered parameter as in \cite{H2013}, \cite{HL2014} and \cite{Ryoko2016} since the true time-varying parameter does not even exist. The requirement that the filtered parameter asymptotically does not have to depend on the arbitrary chosen initialization is very intuitive as otherwise different initializations could provide different results.  

We emphasize that situations of correctly-specified non-invertible models can be thought of as a particular case of misspecification. This interpretation is valid because, under non-invertibility, the true parameter value $\theta_0$ is such that $E\log\Lambda_0(\theta_0)\ge0$ and therefore asymptotically outside the parameter region $\hat \Theta_n$ with probability 1. In such situations, indeed,  the ML estimator constrained on the empirical region $\hat \Theta_n$ is inconsistent with respect to $\theta_0$ but we  can ensure that asymptotically the initialization is not affecting the parameter estimate.

\section{Confidence bounds for the unfeasible parameter region}
\label{sec5}

For a given sample $\{y_1,\dots,y_n\}$, the empirical region $\hat \Theta_n$ may not satisfy the required Lyapunov condition. Therefore, it may be of interest to test whether a point $\theta\in\bar\Theta$ satisfies the invertibility condition. Proposition \ref{pr.ci} establishes the asymptotic normality of the test statistic $T_{n}$ under the null hypothesis that  $H_{0}: E\log \Lambda_0(\theta)=0$. Furthermore, we show that the statistic diverges under the alternative $H_{1}: E\log \Lambda_0(\theta)\neq 0$. This result can naturally be used to produce confidence bounds. Below we let  $\sigma^{2}_{n}$ denote the variance of $ n^{-\frac{1}{2}}\sum_{t=1}^n\log \Lambda_t(\theta)$
\begin{proposition}\label{pr.ci}
Let $\{y_t\}_{t \in \mathbb{Z}}$ be stationary  and geometrically $\alpha$-mixing with  $E|\log \Lambda_0(\theta)|^r<\infty$ for any $\theta \in \bar \Theta$ and $r>2$. Then, under the null hypothesis $H_{0}: E\log \Lambda_0(\theta)=0$ we have
$$T_n \ :=   \ \frac{n^{-\frac{1}{2}}\sum_{t=1}^n\log \Lambda_t(\theta)}{\hat\sigma_n} \quad \xrightarrow{d} \quad N(0,1),\qquad \text{as} \ n \to \infty,$$
where $\hat\sigma_n^2$ is a consistent estimator of $\sigma^{2}_{n}$. Furthermore, $T_n \to -\infty$ as $n \to \infty$ when $E\log \Lambda_0(\theta)<0$, and $T_n \to \infty$ as $n \to \infty$ when $E\log \Lambda_0(\theta)>0$. 
\end{proposition}
The variance $\sigma_n^2$ can be  consistently estimated using the Newey-West  estimator; see \cite{Newey1987}. Proposition \ref{pr.ci} shows that, for any given $\theta$ and at any given confidence level $\alpha$, we ascertain that the test statistic $T_n$ is asymptotically standard normal, if $\theta$ is a boundary point satisfying $E\log\Lambda(\theta)=0$. If the null hypothesis is rejected with negative values of $T_{n}$, then the evidence suggests that the contraction condition is satisfied for that $\theta$, i.e.~that $E\log\Lambda(\theta)<0$. If the null hypothesis is rejected with positive values of $T_{n}$, then the evidence suggests that   $E\log\Lambda(\theta)>0$.  On the basis of  the asymptotic result in Proposition \ref{pr.ci}, we can also obtain level $\alpha$ confidence sets for  $\Theta_0=\left\{\theta \in \bar\Theta: E\log\Lambda_0(\theta)<0\right\}$. More specifically, we consider the set $\hat\Theta_\alpha^{up}=\left\{\theta \in \bar\Theta: T_n<z_{1-\alpha}\right\}$ such that and  for any $\theta\in \Theta_0$ we have
$$\lim_{n\rightarrow\infty}P\{\theta\in\hat\Theta_\alpha^{up}\}\ge 1-\alpha.$$
This means that  any element in the set $\Theta_0$ has   an asymptotic probability of  at least $1-\alpha$ of being contained in the  set $\hat \Theta_\alpha^{up}$. 
Similarly, we  also consider  the set $\hat\Theta_\alpha^{lo}=\left\{\theta \in \bar\Theta: T_n<z_{\alpha}\right\}$ and for this set for that any $\theta \in \Theta_0^c$, where $\Theta_0^c=\left\{\theta \in \bar\Theta:  E\log\Lambda_0(\theta)\ge0 \right\}$, we have that
$$\lim_{n\rightarrow\infty}P\{\theta\in\hat\Theta_\alpha^{lo}\}\le \alpha.$$
The set $\hat\Theta_\alpha^{lo}$ can be viewed as a lower bound confidence set of level $\alpha$ for $\Theta_0$, because it is a conservative set in the sense that we fix the maximum asymptotic probability $\alpha$ such that a $\theta$ not being contained in $\Theta_0$ can be in $\hat\Theta_\alpha^{lo}$. In an equivalent way, the set $\hat\Theta_\alpha^{up}$ can be viewed as an upper bound confidence set for $\Theta_0$. In this case, the maximum asymptotic probability of having an element $\theta\in \Theta_0$ not being in $\hat\Theta_\alpha^{up}$ is fixed at a level $\alpha$.



\section{Some practical examples}
\label{sec6}

\subsection{Beta-$t$-GARCH model}

Consider first the properties of the  Beta-$t$-GARCH model as a data generating process. The basic dynamic process equation in (\ref{tgass}) with $\theta=\theta_0$ can alternatively be expressed as
\[
f^o_{t+1}=\omega_0+ f^o_{t}c_t, \qquad c_t=\beta_0+(\alpha_{0} +\gamma_{0}d_t)(v_0+1)b_t,
\]
where $b_t={\varepsilon_t^2}/({v_0-2+\mathcal{\varepsilon}_t^2})$ has a beta distribution with parameters $1/2$ and $v_0/2$, see Chapter 3 of \cite{H2013}. 
In order to ensure that  $f_t^o$ is positive with probability 1 and that $f_t^o$ is the conditional variance of $y_t$ given $y^{t-1}$, the  parameter vector $\theta_0 =(\omega_0,\beta_0,\alpha_0,\gamma_0,v_0)^T$ has to satisfy the following conditions $\omega_0>0$, $\beta_0\ge0$, $\alpha_0>0$ and $\gamma_0 \ge-\alpha_0$. Letting $v_0\rightarrow\infty$, the Student's $t$ distribution approaches the Gaussian distribution and the recursion of $f_t^o$ in (\ref{tgass}) becomes 
$$f_{t+1}^o=\omega_0+\beta_0 f_t^o+(\alpha_0 +\gamma_0 d_t)y_t^2,$$ such that, in this limiting case of $v_0\rightarrow\infty$, the model reduces to the so-called GJR-GARCH model of \cite{GLOSTEN1993}, and to the GARCH(1,1) model, when $\gamma_0 =0$.

\begin{theorem}\label{SEth}
The model  in (\ref{tgass}) admits a unique stationary  and ergodic solution $\{f_{t}^{o}\}_{t \in \field{Z}}$ if and only if $E\log  c_t <0$. 
\end{theorem}

Theorem \ref{SEth} above derives a necessary and sufficient moment condition for the Beta-$t$-GARCH model to generate stationary ergodic paths. A simpler restriction on the parameters of the model that is sufficient for obtaining stationary and ergodic paths is  
$$\beta_0+\alpha_{0} +\gamma_{0}/2<1.$$
 
Theorem \ref{moments} complements Theorem \ref{SEth} by providing additional restrictions which ensure that the paths generated by the Beta-$t$-GARCH are not only strictly stationary and ergodic, but also have a bounded moment. 

\begin{theorem}\label{moments}
Let $E c_t^z<1$, where $z\in \mathbb{R}^+$, then  (\ref{tgass}) admits a unique stationary  and ergodic solution $\{f_{t}^{o}\}_{t \in \field{Z}}$ that satisfies $E|f_t^o|^z<\infty$.
\end{theorem}

Having analyzed some properties of the Beta-$t$-GARCH as a data generating process, we now turn to the properties of the model as a filter that is fitted to the data. 

\subsubsection*{Invertibility of the filter}
Let us analyze invertibility of the functional filtered parameter $\hat f_t$. 
The filtered equation of the Beta-$t$-GARCH is given by 
\begin{eqnarray}\label{fpbt}
\hat f_{t+1}( \theta)=\omega+\beta \hat f_t(\theta)+(\alpha +\gamma d_t)\frac{(v+1)y_t^2}{(v-2)+ y_t^2/\hat f_t(\theta)}, \qquad t\in \mathbb{N},
\end{eqnarray}
where the recursion is  initialized at a point $\hat f_{0}(\theta)\in \mathcal{F}_\theta=[\bar \omega,\infty)$. The observations $\{y_1,\dots,y_n\}$ are considered to be a realization from a random process. If we assume correct specification, then the generating process is given by (\ref{tgass}) and there exists some true unknown parameter $\theta_{0}$ that defines the properties of the data. It is  straightforward to see that the set $\mathcal{F}_\theta$ where the SRE in (\ref{fpbt}) lies is given by $[\bar \omega,\infty)$. This is true irrespective of the correct specification of the model as 
the last summand on the right hand side of the equation in (\ref{fpbt}) is positive with probability 1.

Corollary \ref{fpt} follows immediately from Proposition \ref{pp1} and  provides sufficient conditions for the desired invertibility result.

\begin{corollary}\label{fpt}
Let $\{y_t\}_{t\in \mathbb{N}}$ be a stationary and ergodic sequence of random variables, and let $\Theta$ be a compact set such that 
$$E\log\left|\beta+(\alpha+\gamma d_0)\frac{(v+1)y^4_0}{\left ((v-2)\bar \omega+y_0^2\right)^2}\right|<0, \; \forall \;\theta \in \Theta,$$
where $\bar \omega = \omega/(1-\beta)$.
Then, the sequence $\{\hat f_t\}_{t\in\mathbb{N}}$ defined in (\ref{fpbt}) is continuously invertible, i.e.
$$\|\hat f_t -\tilde{f}_t\|_\Theta \xrightarrow{\text{e.a.s.}} 0\;\; \text{as} \;\; t\rightarrow \infty,$$
for any initialization $\hat f_0 \in \mathbb{C}(\Theta, \mathcal{F}_\Theta)$ and where $\{\tilde{f}_t\}_{t\in \mathbb{Z}}$ is a stationary and ergodic sequence.
\end{corollary}

It is clearly implied by Corollary \ref{fpt} that the Lipschitz coefficient  $\Lambda_0(\theta)$ depends on the data generating process through $y_0$. Therefore, in practice, the parameter region $\Theta$ cannot be explicitly obtained from the contraction condition $E\log \Lambda_0(\theta)<0$. As we have discussed in Section \ref{sec2}, under the assumption of correct specification or of $y_0$ having a symmetric distribution around zero, the unfeasible contraction condition $E\log \Lambda_0(\theta)<0$ is ensured by the following feasible sufficient condition
\begin{eqnarray}\label{cc}
\frac{1}{2}\log\left|\beta+\alpha (v+1)\right|+\frac{1}{2}\log\left|\beta+(\alpha+\gamma)(v+1)\right|<0.
\end{eqnarray} 
This result is obtained from taking the supremum over $y_0$ from which it follows with probability 1 that
$$E\log\left|\beta+(\alpha+\gamma d_0)\frac{(v+1)y^4_0}{\left ((v-2)\bar \omega+y_0^2\right)^2}\right| \le E\log\left|\beta+(\alpha+\gamma d_0)(v+1)\right|.$$
Then by assuming that the median of $y_0$ is equal to zero, the feasible condition in (\ref{cc}) follows immediately.

The theory developed in Sections \ref{sec3} and \ref{sec4} can be used to formulate an alternative to (\ref{cc}). The estimated  region $\hat\Theta_n$ that satisfies an empirical version of $E\log \Lambda_0(\theta)<0$ is given by
\begin{eqnarray}\label{ec}
n^{-1}\sum_{t=1}^n\log\left|\beta+(\alpha+\gamma d_t)\frac{(v+1)y^4_t}{\left ((v-2) \bar\omega+y_t^2\right)^2}\right|<0.
\end{eqnarray} 
This empirical condition imposes weaker restrictions on the parameter region. In the following, we discuss how the difference between the condition (\ref{cc}) and (\ref{ec}) can be relevant in practice. Figure \ref{fig:reg2} complements Figure \ref{fig:reg1} by showing that our empirical region is significantly larger than the region obtained from (\ref{cc}). Most importantly, Figure \ref{fig:reg2} reveals that the ML point estimates obtained from the S\&P 500 index lie well inside the empirical region. 
\begin{frame}{}
\begin{figure}[h!]
\center
\includegraphics[scale=0.75]{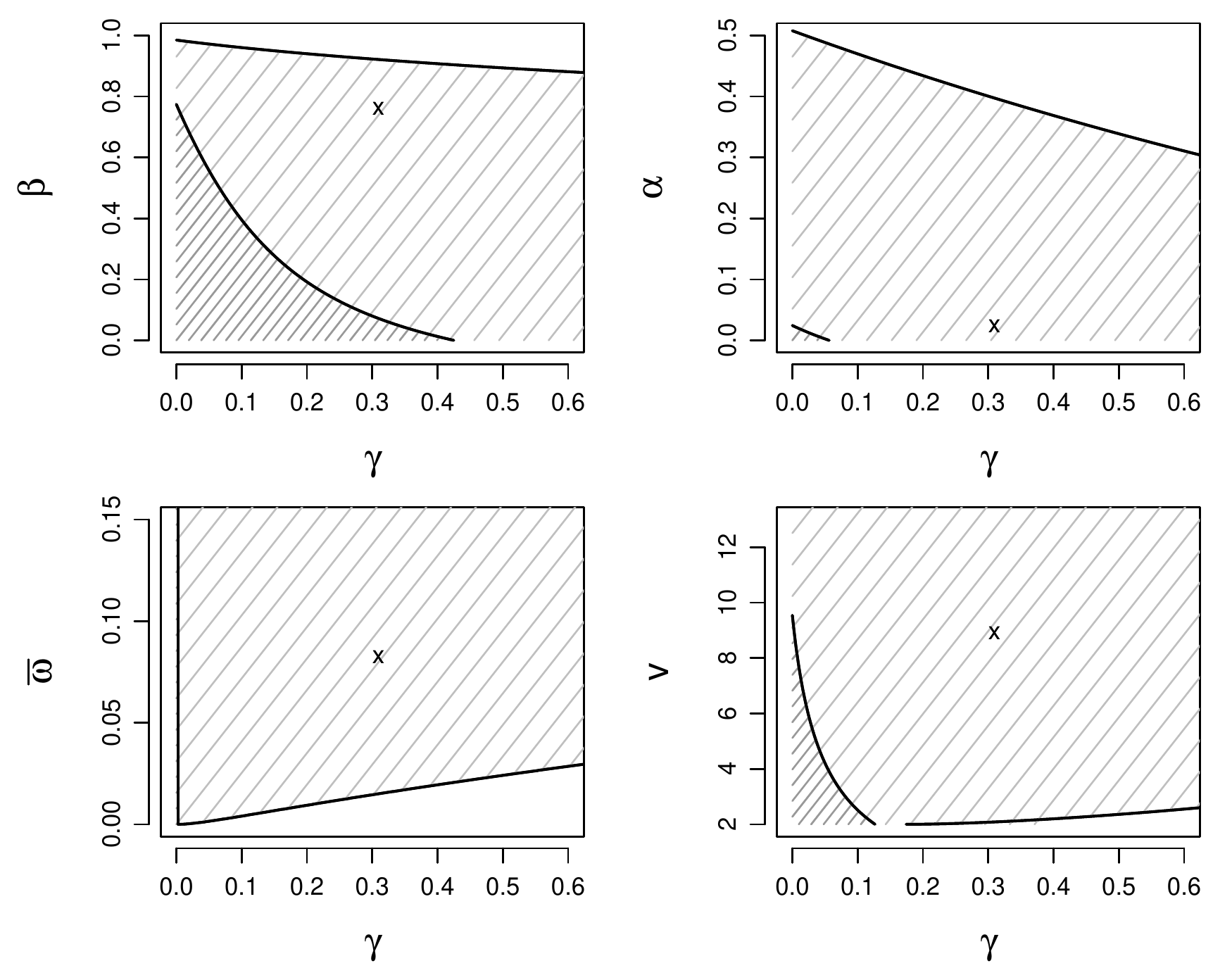}
\vspace{-0.2cm}
\caption{\textit{The light gray area represent the parameter region obtained from (\ref{ec}) for the log-returns of the S\&P 500. In the 2-dimensional plots the other parameters are fixed at their estimated value. The dark gray area is the region obtained from  (\ref{cc}). The crosses denote the estimated value of the parameter.}}
\label{fig:reg2}
\end{figure}
\end{frame}

From the theory developed in Section \ref{sec5}, we obtain the confidence bounds for the unfeasible parameter region. The conditions required for Proposition \ref{pr.ci}, and hence for obtaining the confidence bounds, are valid as can easily be verified in this case. In particular, the condition $E|\log \Lambda_0(\theta)|^r<\infty $ is satisfied for any $r>0$ as long as $\beta>0$. Also, from the results in \cite{francq2006}, it follows that the strong mixing assumption is always satisfied when the model is correctly specified. Figure \ref{fig:reg3} provides  a high degree of  confidence that the  Beta-$t$-GARCH filter is indeed invertible. Figure \ref{fig:reg2} presents the 95\% confidence bounds for the invertibility region. We highlight that the point estimate lies well inside the 95\% lower bound.  

\begin{frame}{}
\begin{figure}[h!]
\center
\includegraphics[scale=0.75]{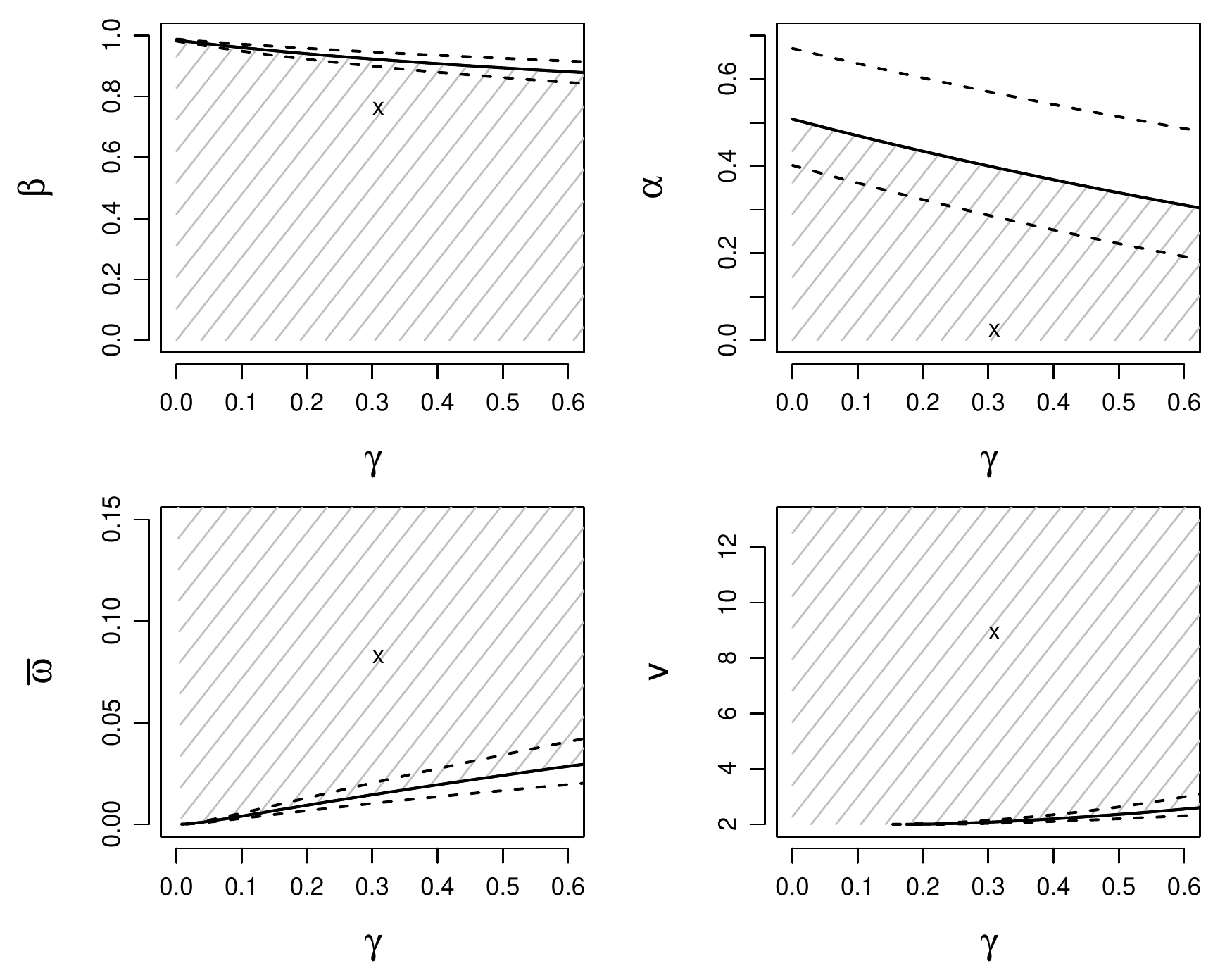}
\vspace{-0.2cm}
\caption{\textit{95\% confidence bounds for the invertibility region are marked by the dashed lines. The light gray area represent the parameter region obtained from (\ref{ec}) for the log-returns of the S\&P 500.  Crosses denote the estimated value of the parameter.} }
\label{fig:reg3}
\end{figure}
\end{frame}

Table \ref{tab:1} reveals that the importance of our empirical invertibility condition is not specific to the S\&P 500 index only. For the monthly time series of financial returns of the well-known indexes considered in Table \ref{tab:1}, we obtain the maximizer $\hat\theta_{n}$ of the likelihood function and we show that inequality (\ref{cc}), evaluated at $\theta=\hat\theta_{n}$, fails whereas inequality (\ref{ec}) holds. These results suggest that condition (\ref{cc}) is too restrictive in practice and that condition (\ref{ec}) can be used to define a reasonably large region of the parameter space on which we can maximize the log-likelihood function.  The last column of Table \ref{tab:1} indicates that the null hypothesis of whether the point estimate is a boundary point of the invertibility region is strongly rejected in all cases.

\begin{table}[ht]
\centering
\begin{tabular}{lcccccccc}
  \cline{2-9}
\vspace{-0.4cm}\\
 & $\omega$ & $\beta$ & $\alpha$ & $\gamma$ & $v$ &  (\ref{cc})& (\ref{ec})& p-value \\ 
  \hline
\vspace{-0.4cm}\\
DJIA&0.058 &0.554& 0.000 & 0.371 & 7.417 & 0.357  & -0.507& 0.000 \\ 
\vspace{-0.6cm}\\
 & \footnotesize{(0.019)} & \footnotesize{(0.160)} & \footnotesize{(0.047)} &\footnotesize{(0.116)}&\footnotesize{(2.339)} & &   & \\ 
\vspace{-0.4cm}\\
S\&P 500 & 0.020 & 0.759 & 0.023 & 0.309 & 8.893 & 0.691 & -0.181 &0.000\\ 
\vspace{-0.6cm}\\
 & \footnotesize{(0.013)} & \footnotesize{(0.114)} & \footnotesize{(0.046)} &\footnotesize{(0.111)}&\footnotesize{(2.640)} & &   & \\ 
\vspace{-0.4cm}\\
NASDAQ  & 0.026  & 0.754 & 0.106 & 0.198 &9.865&1.022&-0.109  &0.000\\ 
\vspace{-0.6cm}\\
 & \footnotesize{(0.010)} &\footnotesize{(0.077)} &\footnotesize{(0.033)} &  \footnotesize{(0.071)} & \footnotesize{(3.396)} &  &   &\\ 
\vspace{-0.4cm}\\
NI 225&0.088 &0.637& 0.000 & 0.230 & 26.552 &  0.746  & -0.416 & 0.000 \\ 
\vspace{-0.6cm}\\
 & \footnotesize{(0.010)} & \footnotesize{(0.000)} & \footnotesize{(0.010)} & \footnotesize{(0.037)} & \footnotesize{(1.083)}& &   & \\ 
\vspace{-0.4cm}\\
FTSE 100 & 0.042 &0.595 & 0.059 & 0.332 &7.621 & 0.737 & -0.378 &0.000\\ 
\vspace{-0.6cm}\\
 & \footnotesize{(0.012)} & \footnotesize{(0.134)} & \footnotesize{(0.049)} &\footnotesize{(0.107)}&\footnotesize{(2.255)} & &   & \\ 
\vspace{-0.4cm}\\
DAX  & 0.046  &0.731 & 0.050 & 0.212 &7.932 &0.642&-0.218  &0.000\\ 
\vspace{-0.6cm}\\
 & \footnotesize{(0.013)} &\footnotesize{(0.088)} &\footnotesize{(0.046)} &  \footnotesize{(0.073)} & \footnotesize{(2.905)} &  &   &\\ 
\hline\vspace{-0.3cm}\\
\end{tabular}
\vspace{-0.2cm}
\caption{\textit{Parameter estimates for the model specified in (\ref{tgass}) for the log-returns of some of the stock indexes Dow Jones Industrial (DJIA), Standard and Poor's 500 (S\&P 500), NASDAQ, Nikkei 225 (NI 225), London Stock Exchange (FTSE) and German DAX. For all these indexes, time series of monthly returns from January 1980 to April 2016 are considered. The columns labeled (\ref{cc}) and (\ref{ec}) contain the values of respectively condition (\ref{cc}) and (\ref{ec}) evaluated at the estimated parameter value. The last column contains the $p$-value of the test whether the point estimate is in a boundary point of the ``true'' invertibility region.} }
\label{tab:1}
\end{table}

Having provided strong evidence of the invertibility of the Beta-$t$-GARCH filter, we are now ready to discuss consistency of the ML estimator in these larger parameter spaces defined by the feasible empirical parameter restrictions.

\subsubsection*{Consistency of the  ML estimator}

The log-likelihood function $\hat L_n$ is defined as in (\ref{llik}) with $\hat l_t(\theta)$  given by
\begin{eqnarray*}
\hat l_t( \theta)= \log \left(\frac{\Gamma\left(2^{-1}(v+1)\right)}{\sqrt{(v-2) \pi}\Gamma\left(2^{-1}v\right)}\right) - \frac{1}{2}\log \hat f_t(\theta)-\frac{v+1}{2}\log \left(1+\frac{y_t^2}{(v-2) \hat f_t(\theta)}\right),
\end{eqnarray*}
where $\Gamma$ denotes the gamma function. Next we obtain the consistency results for the Beta-$t$-GARCH model. The first result follows from an application of Theorem \ref{th1}.
\begin{theorem}\label{CON1}
Let the observed data be generated by a stochastic process $\{y_t\}_{t\in\mathbb{Z}}$ that satisfies the model equations in (\ref{tgass}) at $\theta=\theta_0\in \Theta$ and   let $\Theta $ be a compact  set that satisfies  the condition in (\ref{cons}) and  such that $\omega>0$, $\beta\ge0$, $\alpha\ge0$ , $\gamma\ge-\alpha$ and $v > 2$ for any $\theta \in \Theta$. Then the ML estimator $\hat \theta_n$ defined in (\ref{mle}) is strongly consistent.
\end{theorem} 
Theorem \ref{CON1} considers a more general model but is also extends the asymptotic results in \cite{Ryoko2016} in several directions. In particular, Theorem \ref{CON1} does not impose the assumption that the time-varying parameter   $f_t^o$ is observed at $t=0$. Furthermore, it does not rely on the condition that the likelihood function is maximized on an arbitrarily small neighbourhood around the true parameter $\theta_0$. 
The next result shows the consistency of the ML estimator  in (\ref{smle}) for the Beta-$t$-GARCH model. 
\begin{theorem}\label{CON2}
Let the observed data be generated by a stochastic process $\{y_t\}_{t\in\mathbb{Z}}$ that satisfies the model equations in (\ref{tgass}) at $\theta_0\in \Theta_\delta$ and   let  $\bar \Theta $ be a  compact  set such that $\omega>0$, $\beta>0$, $\alpha\ge0$ , $\gamma\ge-\alpha$ and $v>2$  for any $\theta \in \bar \Theta$. Then the ML estimator $ {\doublehat \theta}_n$ defined in (\ref{smle}) is strongly consistent.
\end{theorem}
In contrast to Theorem \ref{CON1}, Theorem \ref{CON2} does not require the unfeasible invertibility  condition in (\ref{cons}) to be satisfied as the optimization of the likelihood is in a region that satisfies an empirical version of (\ref{cons}).

\subsection{Autoregressive model with time-varying coefficient}

The practical relevance of the empirical invertibility conditions is not restricted to volatility models only. On the contrary, it applies to the general class of observation driven models. Consider the first-order autoregressive model with a time-varying autoregressive coefficient and with a fat-tailed distribution as discussed in \cite{nlgas2014} and \cite{DMP204}. This model is specified by the equations
\begin{eqnarray*}
y_t&=&f_ty_{t-1}+\sigma \varepsilon_t,\qquad \qquad \qquad \qquad \qquad \{\varepsilon_t\}\sim t_v,\\
f_{t+1}&=&\omega +\beta f_{t}+\alpha \frac{ (y_t-f_t y_{t-1})y_{t-1}}{1+v^{-1}\sigma^{-2}(y_t - f_t y_{t-1})^2},
\end{eqnarray*}
where $\sigma$, $\omega$, $\beta$, $\alpha$ and $v$ are static parameters that need to be estimated and $t_v$ denotes the Student's $t$ distribution with $v$ degrees of freedom. This model is not exactly of the form in (\ref{model}) as the conditional density of $y_t$ given $f_t$ depends also on the lagged value $y_{t-1}$. However, the extensions of our results required for including this case, and also possibly exogenous variables in the conditional density, are trivial.

This autoregressive model implies a time-varying autocorrelation function. In particular, it can describe time series that exhibit periods of strong temporal persistence, or near-unit-root dynamics, and periods of low dependence, or strong mean reverting behaviour. There is evidence that various time series in economics feature such complex nonlinear dynamics; see \cite{bec2008acr} for an example in real exchange rates. 
By adopting the results of Proposition \ref{pp1} and taking into account that
$$\dot \phi(f, Y_t^k,\theta) = \beta + \alpha \frac{(y_t-f y_{t-1})^2-v\sigma^2}{\left((y_t-f y_{t-1})^2+v\sigma^2\right)^2}v\sigma^2y_{t-1}^2,$$
we obtain that the stochastic coefficient $\Lambda_t(\theta)$ is given by
$$\Lambda_t(\theta)=\max\left\{|\beta-\alpha y_{t-1}^2|,|\beta + \frac{1}{8}\alpha y_{t-1}^2|\right\}.$$
In this case there is not a clear way to derive sufficient conditions to ensure that $E\log\Lambda_t(\theta)<0$. A trivial solution would impose that $\alpha=0$ and $|\beta|<1$ but in this way we get a degenerate parameter region and $f_t$ becomes a static parameter. This situation is not of practical interest. An alternative option is to rely on the results of Section \ref{sec4} and to estimate the parameter region $\hat\Theta_n$. 

\begin{frame}{}
\begin{figure}[h!]
\center
\includegraphics[scale=0.7]{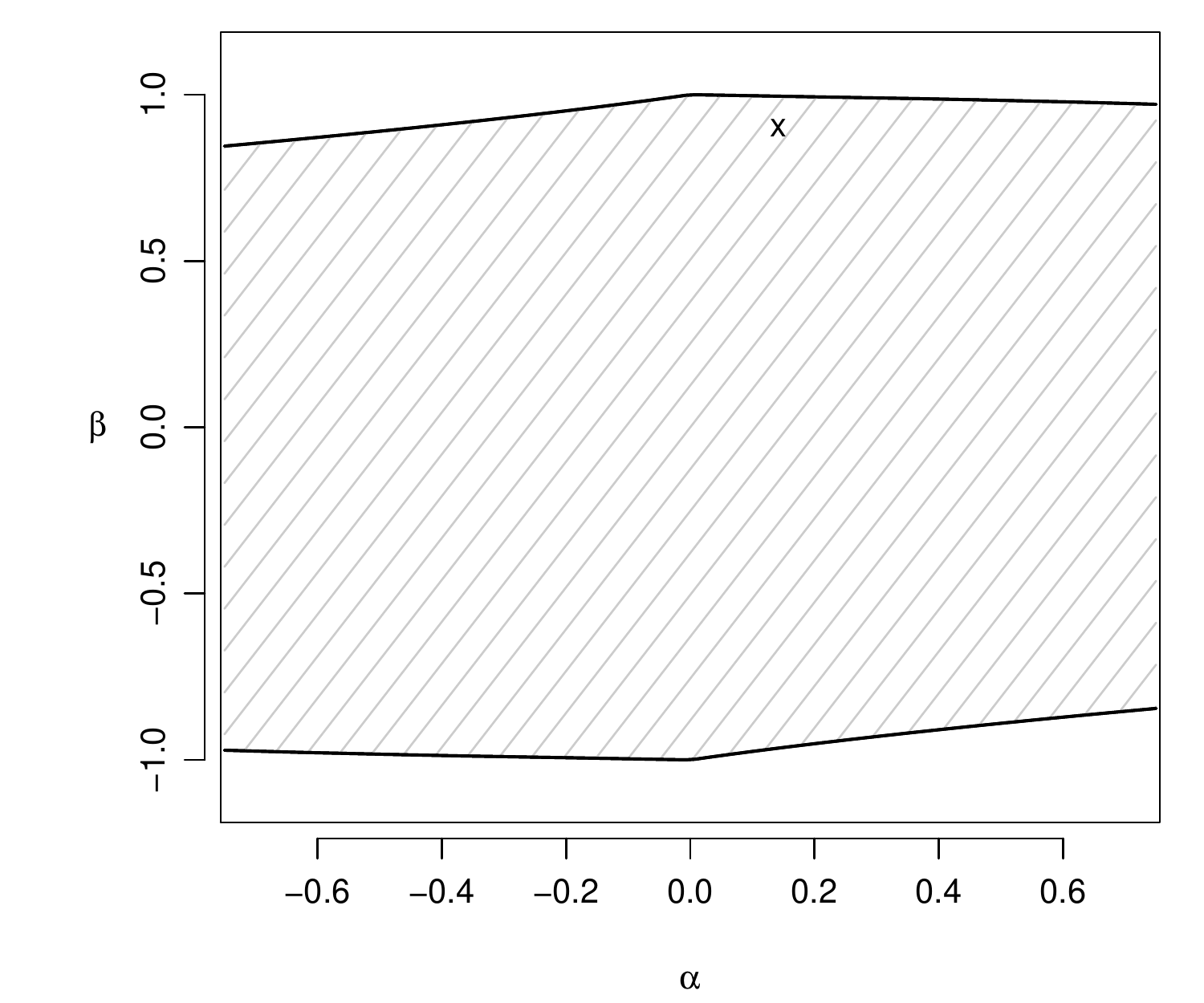}
\vspace{-0.2cm}
\caption{\textit{Parameter region and ML estimate obtained for the autoregressive model with a time-varying autoregressive coefficient 
and applied to the U.S. unemployment claims time series.}}
\end{figure}
\label{fig_3}
\end{frame}

To show how the results of the previous sections can be useful in this situation, we derive the estimated region for the time series of weekly changes of the logarithm of U.S.~unemployment claims; this data set is considered earlier in \cite{nlgas2014}. We analyze this data set using the model given above. From Figure \ref{fig_3} we learn that the maximizer of the likelihood function is contained in the estimated region. This shows how the empirical invertibility condition is not too restrictive. Moreover, due to the results in our study, we can ensure the reliability of the ML estimator.

\subsection{Fat-tailed location model}

Finally, we consider the Student's $t$ location model of \cite{HL2014} which is given by
\begin{eqnarray*}
y_t &=& f_t +\sigma \varepsilon_t,\qquad \qquad \qquad \qquad \qquad \qquad  \{\varepsilon_t\}\sim t_v,\\
f_{t+1} &=& \omega +\beta f_{t}+\alpha \frac{ y_t-f_t}{1+v^{-1}\sigma^{-2}(y_t - f_t )^2},
\end{eqnarray*}
where $\sigma$, $\omega$, $\beta$, $\alpha$ and $v$ are unknown static parameters.
In the application of rail travel data in the United Kingdom, \cite{HL2014}  show that this model is capable of extracting a smooth and robust trend from the rail travel data. \cite{HL2014} also provide an asymptotic theory for the ML estimator of the static parameters of the model. In particular, by relying on Lemma 1 of \cite{JRab2004}, they obtain the ML estimator properties under the restrictive and non-standard assumption that the true time-varying mean at time $t=0$, i.e. $f_{0}^o$, is known. In addition, the asymptotic results derived in \cite{HL2014} are only valid under correct model specification and assuming that the likelihood is maximized on an arbitrarily small parameter space containing $\theta_{0}$.
To complement their results, we address the invertibility issue and obtain new and more general asymptotic results for the ML estimator that do not rely on these restrictive assumptions. 

As long as  $|\beta|<1$, the sequence $\{\hat f_t(\theta)\}$ takes values in $[\bar \omega_l,\bar \omega_u]$, where $\bar \omega_l=(\omega-c)/(1-\beta)$ and $\bar \omega_u=(\omega+c)/(1-\beta)$, with $c=|\alpha|\sqrt{3v\sigma^2}/4$.
Defining the function $s_\theta(x):= v\sigma^2(x^2-v\sigma^2)/(x^2+v\sigma^2)^2$, we obtain that the stochastic coefficient $\Lambda_t(\theta)$ is 
$$\Lambda_t(\theta)=\max\left\{|z_{1t}|,|z_{2t}|\right\},$$
where $z_{1t}$ and $z_{2t}$ are respectively given by
\[ z_{1t} =
  \begin{cases}
    \beta-\alpha       & \text{if} \;\;\;y_t \in [\bar \omega_l, \bar \omega_u],\\
    \beta+\alpha \min\left(s_\theta(y_t-\bar \omega_u), s_\theta(y_t-\bar \omega_l)\right)  & \text{otherwise},\\
  \end{cases}
  \;\;\;\;\;\;\;\;\;\;\;\;\;\;\;
\]
and
\[ z_{2t} =
  \begin{cases}
    \beta+\alpha/8      &  \text{if} \;\;\; y_t\pm \sqrt{3v\sigma^2}  \in [\bar \omega_l, \bar \omega_u],\\
    \beta+\alpha \max\left(s_\theta(y_t-\bar \omega_u), s_\theta(y_t-\bar \omega_l)\right)  & \text{otherwise}.\\
  \end{cases}
\]
An upper bound for $\Lambda_t(\theta)$, independent of $y_t$, is then obtained as 
$$\Lambda_t(\theta)\le \max(|\beta-\alpha|,|\beta+\alpha/8|).$$
This condition can be too restrictive. Figure \ref{fig_4} shows yet another example where these restrictive conditions fail to hold while, on the other hand, their empirical counterparts are satisfied. For illustration purposes, we consider the above model for the time series of monthly changes in the U.S. consumer price index from January 1947 to February 2016. We show in Figure \ref{fig_4} that the estimated parameter region is larger and it contains the parameter estimate.

\begin{frame}{}
\begin{figure}[h!]
\center
\includegraphics[scale=0.75]{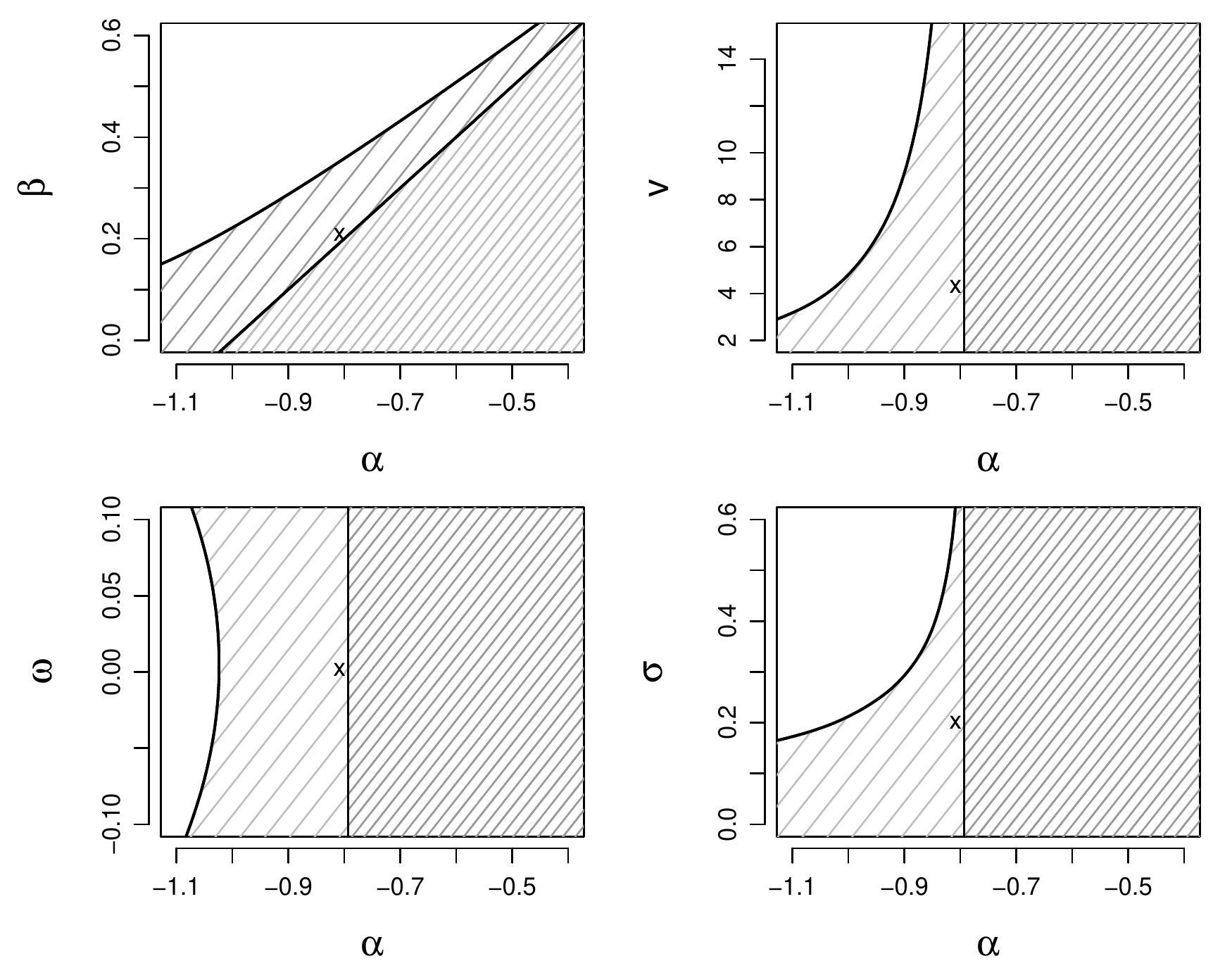}
\vspace{-0.2cm}
\caption{\textit{Parameter region and parameter estimate obtained for the Student's $t$ location model 
 and applied to the U.S. consumer price index time series from January 1947 to February 2016.}}
\end{figure}
\label{fig_4}
\end{frame}


\section{Conclusion}
 \label{sec8}
We have proposed considerably weaker conditions that can be used in practice for ensuring the consistency of the maximum likelihood estimator of the parameter vector in observation-driven time series models. These results are applicable to a wide class of well-known time series models including the generalized autoregressive conditional heteroskedasticity (GARCH) model. Further, we have shown that our consistency results hold for both correctly specified and misspecified models. Finally, we have derived an asymptotic test and confidence bounds for the unfeasible ``true'' invertibility region of the parameter space. The empirical relevance of our theoretical results has been highlighted for a selection of key observation-driven models that are applied to real datasets.

\appendix

\section*{Appendix}

\begin{proof}[Proof of Proposition \ref{pp1}]
To prove this proposition, we first rely on the results of Proposition 3.12 of  \cite{SM2006} and we then employ  the same argument  as in the proof of Theorem 2 of  \cite{Win2013} to relax the uniform contraction condition. This proposition  is closely related to Theorem 2 of  \cite{Win2013}, the main difference is that we explicitly allow the set $\mathcal{F}_\theta$ to depend on $\theta$.

Consider the functional SRE 
 $$\hat f_{t+1}=\Phi_t( \hat f_{t}),\; t\in \mathbb{N},$$
 where the random map $\Phi_t$ is such that $\Phi_t(f)=\phi(f(\cdot), Y_t^k,\cdot)$ for any $f\in\mathbb{C}(C,\mathcal{F}_C)$, where $C$ denotes a compact set. This SRE lies in the separable  Banach space  $\mathbb{C}(C,\mathcal{F}_C)$ equipped with the uniform norm $\|\cdot\|_C$. Therefore,  taking into account that by the mean value theorem
 $$\sup_{f_1,f_2\in\mathcal{F}_C,f_1\neq f_2}\frac{|\phi(f_1, Y_t^k,\theta)-\phi(f_2, Y_t^k,\theta)|}{|f_1-f_2|}\le \sup_{f \in \mathcal{F}_C}| \dot\phi(f, Y_t^k,\theta)|,$$
from Proposition 3.12 of  \cite{SM2006}, it results  that the conditions
\begin{description}
\item[(a)] $E\log^+\|\phi(\bar f, Y_t^k,\cdot)\|_C<\infty \; \text{ for some } \bar f \in \mathcal{F}_C$.
\item[(b)] $E\sup_{\theta \in C}\sup_{f \in \mathcal{F}_C}\log^+|\dot\phi(f, Y_t^k,\theta)|<\infty.$
\item[(c)] $E\sup_{\theta \in C}\sup_{f \in \mathcal{F}_C}\log |\dot\phi(f, Y_t^k,\theta)|<0 .$
\end{description}
 are sufficient to apply Theorem of 3.1 \cite{Bougerol1993} and obtain the convergence result $\|\hat f_t -\tilde f_t\|_{C}\xrightarrow{e.a.s.}0.$ Note that this is  true for any given compact set $C$ that satisfies (a)-(c). 
Now, we define the following stochastic function
$$\Lambda_t^*(\theta_1,\theta_2):=\sup_{f \in \mathcal{F}_{\theta_1}}|\dot \phi(f, Y_t^k,\theta_2)|,$$
and, we define  a compact neighborhood  of $\theta \in \Theta$ with radius $\epsilon >0$ as $B_\epsilon(\theta)=\{\tilde \theta \in \Theta: \|\theta-\tilde \theta\|\le\epsilon\}$. Then, for any non-increasing sequence of constants $\{\epsilon_i\}_{i \in \mathbb{N}}$ such that $\lim_{i\rightarrow \infty}\epsilon_i=0$, the sequence $\left\{\sup_{(\theta_1,\theta_2)\in B_{\epsilon_i}(\theta)\times B_{\epsilon_i}(\theta)}\log \Lambda_0^*(\theta_1,\theta_2)\right\}_{i \in \mathbb{N}}$ is a non-increasing sequence of random variables and by continuity, which is ensured by (iii), we have that
$$\lim_{i\rightarrow \infty}\sup_{(\theta_1,\theta_2)\in B_{\epsilon_i}(\theta)\times B_{\epsilon_i}(\theta)}\log \Lambda_0^*(\theta_1,\theta_2)=\log \Lambda_0(\theta).$$ Condition (ii) implies that $E\sup_{(\theta_1,\theta_2)\in \Theta \times \Theta}\log \Lambda_0^*(\theta_1,\theta_2)\in \mathbb{R}\cup \{-\infty\}$. As a result, we can apply the monotone convergence theorem and obtain
$$ E\lim_{i\rightarrow \infty}\sup_{(\theta_1,\theta_2)\in B_{\epsilon_i}(\theta)\times B_{\epsilon_i}(\theta)}\log \Lambda_0^*(\theta_1,\theta_2)=E\log \Lambda_0(\theta).$$ Therefore, for any $\theta \in \Theta$ such that $E\log \Lambda_0(\theta)<0$ there exists an $\epsilon_\theta>0$ such that
$$ E\sup_{(\theta_1,\theta_2)\in B_{\epsilon_\theta}(\theta)\times B_{\epsilon_\theta}(\theta)}\log \Lambda_0^*(\theta_1,\theta_2)<0.$$ From this and noting that 
 $$\sup_{\theta \in B_{\epsilon_\theta}(\theta)}\sup_{f \in \mathcal{F}_{B_{\epsilon_\theta}(\theta)}}\log |\dot\phi(f, Y_t^k,\theta)|=\sup_{(\theta_1,\theta_2)\in B_{\epsilon_\theta}(\theta)\times B_{\epsilon_\theta}(\theta)}\log \Lambda_0^*(\theta_1,\theta_2),$$
 we obtain that the  conditions (a)-(c) are satisfied for the compact set $B_{\epsilon_\theta}(\theta)$ as (i) implies (a), (ii) implies (b) and (iii) implies (c). Therefore, we conclude that 
$$\|\hat f_t -\tilde f_t\|_{B_{\epsilon_\theta(\theta)}}\xrightarrow{e.a.s.}0.$$ 
The desired result follows as $\Theta$ is  compact and  $\Theta= \bigcup_{\theta \in \Theta}B_{\epsilon_\theta}(\theta)$. Therefore, there exists a finite set of points $\{\theta_1,\dots,\theta_K\}$ such that $\Theta= \bigcup_{k = 1}^K B_{\epsilon_k}(\theta_k)$ and it follows that
$$\|\hat f_t -\tilde f_t\|_{\Theta}=\bigvee_{k=1}^K\|\hat f_t -\tilde f_t\|_{B_{\epsilon_k(\theta_k)}}\xrightarrow{e.a.s.}0.$$
\end{proof}

\begin{proof}[Proof of Proposition \ref{pp3}]
By a.s.~convergence of $\hat \theta_n$ to $\theta_0$, there exists a random integer $T$ such that $\hat \theta_n \in B_{\epsilon_{\theta_0}}(\theta_0)$ for any $n\ge t\ge T$. Keeping the same notation than in the proof of Proposition \ref{pp1} above, let us define the stationary sequence $\rho_t:=\sup_{(\theta_1,\theta_2)\in B_{\epsilon_{\theta_0}}(\theta_0)\times B_{\epsilon_{\theta_0}}(\theta_0)} \Lambda_t^*(\theta_1,\theta_2)$ so that $E\log \rho_0<\infty$.
For $t> T$, we have
\begin{align*}
|\hat f_{t}(\hat \theta_n)-\tilde f_{t}(\theta_0)|&\le \|\hat f_t-\tilde f_t\|_{B_{\epsilon_{\theta_0}}}+|\tilde f_{t}(\hat \theta_n)-\tilde f_{t}(\theta_0)|\\
&\le \rho_t\|\hat f_{t-1}-\tilde f_{t-1}\|_\Theta+|\tilde f_{t}(\hat \theta_n)-\tilde f_{t}(\theta_0)|\\
&\le \prod_{s=T+1}^t\rho_s\|\hat f_{T}-\tilde f_{T}\|_\Theta+|\tilde f_{t}(\hat \theta_n)-\tilde f_{t}(\theta_0)|.
\end{align*}
The first term of the sum converges a.s. to $0$.
One can focus on the last term of the sum that can be bounded with
$$
\underbrace{|\phi(\tilde f_{t-1}(\hat \theta_n),Y_t^k,\hat \theta_n)-\phi(\tilde f_{t-1}(\hat \theta_n),Y_t^k,\theta_0)|}_{w_t(\hat\theta_n)}+|\phi(\tilde f_{t-1}(\hat \theta_n),Y_t^k,\theta_0)-\phi(\tilde f_{t-1}(\theta_0),Y_t^k,\theta_0)|.
$$
For any $\theta\in \Theta$ we have that
$$
|\phi(\tilde f_{t-1}(\hat \theta_n),Y_t^k,\theta)|\le \sup_{f\in\mathcal F_\Theta}|\dot\phi(f, Y_t^k,\theta)|(\|\tilde f\|_\Theta+|\bar f|)+|\phi(\bar f,Y_t^k,\theta)|.
$$
Conditions (i) and (ii) plus the extra condition $E[\log^+\|\tilde f_0\|_\Theta]<\infty$ ensure the existence of the logarithmic moments of $|\phi(\tilde f_{t-1}(\hat \theta_n),Y_t^k,\theta)|$ for any $\theta\in\Theta$. Thus we have $E\sup_{\theta\in\Theta}\log^+ w_t(\theta)<\infty$. Moreover, thanks to the SRE and for $n\ge t\ge T$ we have 
$$
|\phi(\tilde f_{t-1}(\hat \theta_n),Y_t^k,\theta_0)-\phi(\tilde f_{t-1}(\theta_0),Y_t^k,\theta_0)|\le \rho_t
|\tilde f_{t-1}(\hat \theta_n)-\tilde f_{t-1}(\theta_0)|.
$$
By a recursive argument, we obtain for any $n\ge t\ge T$,
\begin{align*}
|\tilde f_{t}(\hat \theta_n)-\tilde f_{t}(\theta_0)|&\le \rho_t|\tilde f_{t-1}(\hat \theta_n)-\tilde f_{t-1}(\theta_0)|+w_t(\hat\theta_n)\\
&\le \sum_{s\le t}\prod_{k=s+1}^t\rho_k w_s(\hat\theta_n).
\end{align*}
Applying Lemma 2.5.2 of \cite{Straumann2005} under $E\sup_{\theta\in\Theta}\log^+ w_t(\theta)<\infty$, we show the uniform convergence on $\Theta$ of the upper bound. We conclude by a continuous argument that this upper bound tends to 0 as $w_s(\hat\theta_n)\to 0$ a.s.~for any $s\le t \le n$ when $t\to\infty$.
\end{proof}

\begin{proof}[Proof of Theorem \ref{th1}]
We prove the theorem from the following intermediate steps:
\begin{description}
\item[(S1)] The model is identifiable, i.e.  $L(\theta_0)>L(\theta)$ for any $\theta\in \Theta$, $\theta\neq\theta_0$.
\item[(S2)] The function $\hat L_n$ converges a.s.~uniformly to $ L_n$ as $n\xrightarrow{} \infty$, i.e. $\|\hat L_n- L_n\|_\Theta \xrightarrow{\text{a.s.}} 0$ as $n\xrightarrow{} \infty$.
\item[(S3)] For any $\epsilon>0$, the following inequality holds with probability 1
\begin{eqnarray}\label{c1}
\limsup_{n \xrightarrow{}\infty}\sup_{\theta \in  B^c(\theta_0,\epsilon)}\hat L_n(\theta)<L(\theta_0),
\end{eqnarray}
where $ B^c(\theta_0,\epsilon)= \Theta \setminus B(\theta_0, \epsilon)$   with  $B(\theta_0,\epsilon)=\{\theta\in \Theta : \|\theta_0-\theta\|< \epsilon\}$;
\item[(S4)]  The result in (S3) implies strong consistency. 
 \end{description}
\begin{inparaenum}[(S1)]

\item First note that, by \textbf{C1},  $L(\theta_0)$ exists and is finite and, by \textbf{C5}, $L(\theta)$ exists for any $\theta \in \Theta$ with either $L(\theta)=-\infty$ or $L(\theta)\in \mathbb{R}$. For the values $\theta\in \Theta$ such that $L(\theta)=-\infty$, the result $L(\theta_0)>L(\theta)$ follows immediately as $L(\theta_0)$ is finite. Hence, from now on, we consider only the values $\theta\in \Theta$ such that $L(\theta)$ is finite. 
It is well known that $\log(x)\le x-1$ for any $x\in \mathbb{R}^+$ with the equality  only in the case $x=1$. This implies that almost surely
\begin{eqnarray}\label{ineq}
l_0(\theta)-l_0(\theta_0)\le\frac{p(y_0|\tilde f_0(\theta),\theta)}{p(y_0|f_0^o,\theta_0)}-1.
\end{eqnarray}
Moreover, we have that  the inequality in (\ref{ineq}) holds as a strict inequality with positive probability as the possibility that $p(y_0|\tilde f_0(\theta),\theta)=p(y_0|f_0^o,\theta_0)$ a.s.~is ruled out by \textbf{C2} for any $\theta \neq \theta_{0}$. As a result
\begin{eqnarray*}
E\left[E\left[l_0(\theta)-l_0(\theta_0)|y^{-1}\right]\right] < E\left[E\left[\frac{p(y_0|\tilde f_0(\theta),\theta)}{p(y_0|f_0^o,\theta_0)}\Big|y^{-1}\right]\right]-1=0, \quad  \ \forall \ \theta \neq \theta_{0}
\end{eqnarray*}
where the right hand side of the inequality is equal to zero as $p(y_0|f_0^o,\theta_0)$ is the true conditional density function.
The desired result $L(\theta_0)>L(\theta)$ follows as $l_0(\theta)-l_0(\theta_0)$ is integrable and therefore by the law of total expectation
$$L(\theta)-L(\theta_0)=E[E[l_0(\theta)-l_0(\theta_0)|y^{-1}]]<0 \quad \ \forall \ \theta \neq \theta_{0}.$$
This concludes the proof of step (S1).


\item
First, note that $\|\hat f_t-\tilde f_t\|_\Theta \xrightarrow{\text{e.a.s.}} 0$ as $t \to \infty$ by an application of Proposition \ref{pp1} as conditions \emph{(i)-(iii)} hold by  \textbf{C3} and  $\{y_t\}_{t\in \mathbb{Z}}$ is stationary and ergodic by \textbf{C1}. Second, by Lemma 2.1 of  \cite{SM2006} the series $\sum_{t=N}^\infty \eta_t \|\hat f_t-\tilde f_t\|_\Theta $ converges a.s. and therefore the inequality in \textbf{C4} implies $\sum_{t=N}^\infty  \|\hat l_t-l_t\|_\Theta < \infty $ a.s.. As a result  $ n^{-1} \sum_{t=1}^n \|\hat l_t - l_t\|_\Theta \xrightarrow{\text{a.s.}} 0$  and  $\|\hat L_n- L_n\|_\Theta \xrightarrow{\text{a.s.}} 0$  follows as $\|\hat L_n- L_n\|_\Theta \le n^{-1} \sum_{t=1}^n \|\hat l_t - l_t\|_\Theta$ for any $n\in \mathbb{N}$. This concludes the proof of (S2).

\item
First, note that in virtue of (S2) $\hat L_n$ is asymptotically equivalent to $L_n$ and therefore we just need to prove that (S3) holds for $L_n$. To show this, a similar  argument as in the proof of  Lemma 3.11 of \cite{Pfanzagl1969} is employed.
Consider any decreasing sequence of real numbers $\{\epsilon_i\}_{i \in \mathbb{N}}$ such that $\lim_{i  \xrightarrow{}\infty}\epsilon_i=0$, then $\{\sup_{\theta^*\in B(\theta,\epsilon_i)}l_0(\theta^*)\}_{i \in \mathbb{N}}$ defines a non-increasing sequence of random variables and, by continuity, we have that  $\lim_{i  \xrightarrow{}\infty} \sup_{\theta^*\in B(\theta,\epsilon_i)}l_0(\theta^*)=l_0(\theta)$. As \textbf{C5}  implies  $E\sup_{\theta \in \Theta} l_0(\theta)<+\infty$ we can apply the monotone convergence theorem and we get 
$$\lim_{i  \xrightarrow{}\infty} E \sup_{\theta^*\in B(\theta,\epsilon_i)}l_0(\theta^*)=L(\theta).$$
Recalling that $L(\theta_0)>L(\theta)$ by (S1), we have that for any $\theta\neq\theta_0$ there exists an $\epsilon_\theta>0$ such that 
$$\limsup_{n\xrightarrow{}\infty}\sup_{\theta^* \in  B(\theta,\epsilon_\theta)} L_n(\theta^*)\le E \sup_{\theta^*\in B(\theta,\epsilon_\theta)}l_0(\theta^*)<L(\theta_0).$$ 
Finally, by compactness of $B^c(\theta_0,\epsilon)$ and by $B^c(\theta_0,\epsilon)\subseteq \bigcup_{\theta\in B^c(\theta_0,\epsilon) }B(\theta,\epsilon_\theta)$, there is a finite set of points $\{\theta_1,\dots,\theta_K\}$ such that $B^c(\theta_0,\epsilon)\subseteq \bigcup_{k=1}^K B(\theta_k,\epsilon_{k})$. Therefore, for any $n \in \mathbb{N}$ we have
$$\sup_{\theta \in  B^c(\theta_0,\epsilon)} L_n(\theta)\le \bigvee_{k=1}^K n^{-1}\sum_{t=1}^n\sup_{\theta \in  B(\theta_k,\epsilon_k)} l_t(\theta),$$
and taking the limit in both sides of the inequality it results
$$\limsup_{n \xrightarrow{}\infty}\sup_{\theta \in  B^c(\theta_0,\epsilon)} L_n(\theta)\le \bigvee_{k=1}^K E\sup_{\theta \in  B(\theta_k,\epsilon_k)} l_0(\theta)<L(\theta_0).$$ This concludes the proof of (S3).

\item This last step follows from standard arguments due to \cite{wald1949}.
From the definition of the ML estimator, we have $\hat L_n(\hat\theta_n(\hat f_0) )\ge \hat L_n (\theta_0)$ for any $n \in \mathbb{N}$. Therefore, given the result in (S3),  we have that
\begin{eqnarray}\label{ccpp}
\liminf_{n\xrightarrow{}\infty} \hat L_n(\hat\theta_n(\hat f_0) )\ge L (\theta_0).
\end{eqnarray} 
Now, if we assume that there exists an $\epsilon>0$ such that $\limsup_{n\xrightarrow{}\infty}\|\hat\theta_n(\hat f_0) -\theta_0\|\ge \epsilon$, then in virtue of (\ref{cc}) it must hold that
$$\limsup_{n\xrightarrow{}\infty} \sup_{\theta \in  B^c(\theta_0,\epsilon)}\hat L_n(\theta)\ge L (\theta_0),$$ 
but because of (\ref{c1}) this event has probability zero. As a result, $\limsup_{n\xrightarrow{}\infty}\|\hat\theta_n(\hat f_0) -\theta_0\|< \epsilon$ with probability 1 for any $\epsilon>0$. This concludes the proof of the theorem.
\end{inparaenum}
\end{proof}


\begin{proof}[Proof of Theorem \ref{th2}]
To prove this theorem we show that the steps (S1)-(S4) in the proof of Theorem \ref{th1} hold replacing the set $\Theta$ with the set $\hat \Theta_n$.

First we show that the  following results hold true

\begin{description}
\item[(a)] Almost surely, for large enough $n$, the true parameter vector $\theta_0$ is  contained in the set $\hat\Theta_n$.
\item[(b)] Almost surely, for large enough $n$, the set $\hat\Theta_n$ is  contained in the compact set $\Theta_\delta$ defined as $\Theta_{\delta/2}:=\{\theta \in \bar \Theta: E\log \Lambda_0(\theta)\le-\delta/2\}$.
 \end{description}
By the a.s.~continuity of  $ \log \Lambda_t(\theta)$  in $\bar\Theta$ ensured by \textbf{A2}, the sequence $\{\log \Lambda_t\}_{t \in \mathbb{N}}$ is a stationary and ergodic sequence of elements in the separable Banach space $\mathbb{C}(\bar\Theta,\mathbb{R})$ equipped with the uniform norm $\|\cdot\|_{\bar\Theta}$. The uniform integrability condition $E\|\log\Lambda_0\|_{\bar\Theta}<\infty$ in \textbf{A2} allows to apply the ergodic theorem of \cite{rao1962} and it follows that 
\begin{eqnarray}\label{unicon}
\left\|n^{-1}\sum_{t=1}^n \log \Lambda_t-E\log \Lambda_0\right\|_{\bar\Theta}\xrightarrow{\text{a.s.}} 0,\; n \xrightarrow{} \infty.
\end{eqnarray}
This implies that  for a large enough $n$ all the points $\theta\in \bar\Theta$ such that $E\log\Lambda_0(\theta)<-\delta$ are contained in $\hat\Theta_n$. Therefore, the result (a) holds as condition \textbf{A1} ensures that $E\log\Lambda_0(\theta_0)<-\delta$. As concerns the result (b), the application of the uniform ergodic theorem implies that the map $\theta \mapsto E\log \Lambda_0(\theta)$ is continuous in $\bar \Theta$. This yields that the set  $\Theta_{\delta/2}$ is compact. Finally,  $\hat\Theta_n \subset \Theta_{\delta/2}$ almost surely for large enough $n$ follows immediately from (\ref{unicon}).

 Indeed, $\Theta_{\delta/2}$ is a compact set contained in $\bar \Theta$ and such that $E \log \Lambda_0(\theta)<0$ for any $\theta \in \Theta_{\delta/2}$. Therefore, from the result (b) together with \textbf{A1}-\textbf{A3}, it is easy to see that (S1) is a.s.~satisfied for large enough $n$ as it holds for the set $\Theta_{\delta/2}$. We also have that (S2) and (S3) are satisfied for the set $\hat \Theta_n$ as they hold for the set $\Theta_{\delta/2}$. Finally, the step (S4) follows in the same way as in the proof of Theorem \ref{th1} by noting that (a) implies that 
 $$\hat L_n(\doublehat\theta_n(\hat f_0) )\ge \hat L_n (\theta_0)$$
almost surely for large enough $n$.
 
 \end{proof}

\begin{proof}[Proof of Theorem \ref{th22}]
The expectation $E\log p^o(y_0|y^{-1})$ exists and is finite by \textbf{M1} and moreover $E\log p(y_0|\tilde f_0(\theta),\theta)$ exists for any $\theta \in \Theta_0$  by \textbf{M3}. This implies that the marginal KL divergence $KL(\theta)$ is well defined for any $\theta \in \Theta_0$. The condition \textbf{M2} guarantees that $L(\theta)$ has a unique maximizer in $\Theta_0$, which is denoted by $\theta^*$. This implies that $\theta^*$ is the unique minimizer of the average KL divergence $KL(\theta)$. As concerns the consistency result, replacing $\theta_0$ with $\theta^*$, the proof is equivalent to the the proof of Theorem \ref{th2}. This can be easily seen as the step (S1) holds by assumption replacing $\theta_0$ with $\theta^*$. Then, the steps (S2)-(S4) do not rely on the correct specification of the model and the consistency is obtained with respect to maximizer of the limit function $L$, which in this case is given by $\theta^*$.

\end{proof}

\begin{proof}[Proof of Proposition \ref{pr.ci}]
For any $\theta\in\Theta$, the random coefficient $\Lambda_t(\theta)$ is a measurable function of $Y_{t}^{k}$ for any  given $k\in \mathbb{N}$. Therefore, as $\{y_t\}_{t\in\mathbb{Z}}$ is geometrically $\alpha$-mixing ,   it results that  $\{\log \Lambda_t(\theta)\}_{t \in \mathbb{Z}}$ is  geometrically $\alpha$-mixing as well. Given  the convergence in probability of $\hat \sigma^2_n$ to  $$\lim_{n\rightarrow\infty} \text{Var}\left(n^{-1/2}\sum_{i=1}^n \log \Lambda_t(\theta) \right)$$ and accounting that $E|\log\Lambda_t(\theta)|^r<\infty$, the asymptotic normality result then follows immediately by an application of a central limit theorem for strong mixing processes (see for instance Theorem 7.8 of \cite{durrett2004})   together with an application of Slutsky's theorem.
\end{proof}

\begin{proof}[Proof of Theorem \ref{SEth}]
First note that  the model equation $f_{t+1}^o=\omega_0+f^o_t c_t$ is a stochastic recurrence equation of the form $f_{t+1}^o=\psi_t(f^o_t)$, where $\psi_t(x):=\omega_0+x c_t$ for any $x \in [0,\infty)$. Therefore, $\{\psi_t\}_{t \in \mathbb{Z}}$ is a stochastic sequence of maps from $[0,\infty)$ into $[0,\infty)$. The proof of the if part of the theorem follows noting   that the condition $E\log c_t<0$ is sufficient to satisfy the assumptions of Theorem 3.1 in  \cite{Bougerol1993}. In particular, the first assumption is  satisfied as $E |\omega_0+x c_t|<\infty$ for any $x\in [0,\infty)$ whereas the second assumption immediately holds by $E\log c_t<0$.

As concerns the only if part, we consider a similar argument as in \cite{bougerol1992}. 
In particular,  we show that if $\{f_t^o\}_{t \in \mathbb{Z}}$ is a stationary and ergodic solution of  (\ref{tgass}), then $E\log c_t$ has to be strictly negative. From the recursion
$$f_t^o=\omega_0\left(1+\sum_{k=1}^{n-1}\prod_{i=1}^kc_{t-i}\right)+\prod_{i=1}^nc_{t-i}f^o_{t-n},$$
it follows that almost surely the following inequality holds
$$\sum_{k=1}^{n-1}\prod_{i=1}^{k}c_{t-i}\le f^o_t, \;\;\;\; \forall \; n\in  \mathbb{N}.$$
This means that $\lim_{n \to \infty}\sum_{k=1}^{n-1}\prod_{i=1}^{k}c_{t-i}$ has to be finite almost surely and therefore $\prod_{i=1}^{k}c_{t-i}$ has to converge almost surely to zero as $k \rightarrow \infty$. As $\{c_t\}_{t \in \mathbb{Z}}$ is an i.i.d sequence of random variables, the almost sure convergence to zero of $\prod_{i=1}^{k}c_{t-i}$ implies that $E\log c_t$ is strictly negative by lemma 2.1 of \cite{bougerol1992}. This concludes the proof of the theorem.
\end{proof}

\begin{proof}[Proof of Theorem \ref{moments}]

When the process admits a stationary solution, the following representation holds
$$f_t^o=\omega_0\left(1+\sum_{k=1}^{\infty}\prod_{i=1}^kc_{t-i}\right).$$ 
In the case $z\in [1,\infty)$, by the Minkowski inequality and considering that $\{c_t\}_{t\in \mathbb{Z}}$ is an i.i.d. sequence of positive random variables, we have that
$$\left(E(f_t^o)^z\right)^{1/z}\le\omega_0\left(1+\sum_{k=1}^{\infty} (E c_{t-i}^z)^{k/z}\right).$$ 
Therefore, when $E c_{t-i}^z<1$, the result $E(f_t^o)^z<\infty$ follows from the convergence of the series $\sum_{k=1}^{n} (E c_{t-i}^z)^{k/z}$. As concerns the case $z\in [0,1)$, by sub-additivity we have that
$$E(f_t^o)^z\le\omega_0^z\left(1+\sum_{k=1}^{\infty} (E c_{t-i}^z)^{k}\right).$$  Then, as before, the desired result follows from the convergence of the series $\sum_{k=1}^{n} (E c_{t-i}^z)^{k}$.
\end{proof}

\begin{proof}[Proof of Theorem \ref{CON1}]
 First note that the expression of the probability density function of a Student's $t$ random variable with $v$ degrees of freedom is
$$k_v(x)=s(v)(1+v^{-1}x^2)^{-(v+1)/2},$$
where
$$s(v)=\frac{\Gamma\left(2^{-1}(v+1)\right)}{\sqrt{v\pi} \Gamma\left(2^{-1}v\right)},$$
and where $\Gamma$ denotes the gamma function.

In the following we check that the conditions C1-C5 are satisfied, then the proof  follows by an application of Theorem  \ref{th1}.

(C1) The stationarity and ergodicity of the sequence $\{y_t\}_{t\in\mathbb{Z}}$ is a direct consequence of Theorem \ref{SEth}. In the following, we prove that the integrability condition $E|l_0(\theta_0)|\le\infty$ is satisfied. First, note that $l_0(\theta_0)$ is given by
$$l_0(\theta_0)=\log s(v_0)-\frac{1}{2}\log f_0^o - \frac{v_0+1}{2}\log \left( 1+v_0^{-1}\varepsilon_0^2\right),$$
therefore we just need to show that $E|\log  f_0^o|<\infty$ holds. Consider a decreasing sequence of numbers $\{\epsilon_i\}_{i \in \mathbb{N}}$, $\epsilon_i >0$, such that $\lim_{i \rightarrow \infty}\epsilon_i=0$, then $\{(c_t^{\epsilon_i}-1)/{\epsilon_i}\}_{i \in \mathbb{N}}$ is a decreasing sequence of random variables such that $\lim_{i \rightarrow \infty}(c_t^{\epsilon_i}-1)/{\epsilon_i}=\log c_t$. An application of the monotone convergence theorem leads to
$$\lim_{i \rightarrow \infty}E\left(\frac{c_t^{\epsilon_i}-1}{\epsilon_i}\right)= E\log c_t.$$
Therefore if $E\log c_t<0$, then there exists an $\bar\epsilon>0$ such that  $E(c_t^{\bar\epsilon}-1)/{\bar\epsilon}<0$ and thus $E c_t^{\bar\epsilon}<1$. In virtue of Theorem \ref{moments},  $E(f_t^o)^{\bar\epsilon}<\infty$ and thus we have that $E \log^+f^o_t<\infty$.  The desired result follows as $f^o_t\ge \omega_0/(1-\beta_0)>0$ a.s. and therefore $E \log^+f^o_t<\infty$ implies $E |\log f^o_t|<\infty$.

(C2)  Note that  $a_1k_{v_1}(a_1 x)=a_2k_{v_2}(a_2 x)$ for any $x\in \mathbb{R}$ if and only if $(v_1,a_1)=(v_2,a_2)$. Therefore, if $\varepsilon_0\sim t_v$ then $a_1k_{v_1}(a_1 \varepsilon_0)=a_2k_{v_2}(a_2 \varepsilon_0)$ a.s. if and only if $(v_1,a_1)=(v_2,a_2)$ as $\varepsilon_0$ is an absolutely continuous random variable with a positive density function on $\mathbb{R}$.   As a result, considering that  $l_0(\theta_0)=l_0(\theta)$ a.s. if and only if 
$$k_{v_0}(\varepsilon_0)=\sqrt{\frac{f_0^o}{\tilde f_0(\theta)}}k_v\left(\sqrt{\frac{f_0^o}{\tilde f_0(\theta)}}\varepsilon_0\right) \; \text{a.s.},$$
we have  that  $l_0(\theta_0)=l_0(\theta)$  a.s. if and only if $v=v_0$ and $f_0^o=\tilde f_0(\theta_0)$ a.s.. This means that  the non-trivial implication $l_0(\theta_0)=l_0(\theta)$ a.s. only if  $\theta = \theta_0$ is satisfied if we can show that, given $v=v_0$, $f_0^o=\tilde f_0(\theta)$ a.s. only if  $\theta = \theta_0$. Considering that the sequence $\{\tilde f_t\}_{t \in \mathbb{Z}}$ is stationary, we have that $f_0^o=\tilde f_0(\theta)$ a.s. is the same as $f_t^o=\tilde f_t(\theta)$ a.s. for any $t \in \mathbb{Z}$. Assuming $f_{t}^o=\tilde f_{t}(\theta)$ a.s., the difference $f_{t+1}^o - \tilde f_{t+1}(\theta)$ satisfies
$$f_{t+1}^o -\tilde f_{t+1}(\theta)=\omega_0-\omega + f_{t}^o z_{t},$$
$$z_{t}=\beta_0-\beta +\Big(\alpha_{0}-\alpha+(\gamma_{0}-\gamma)d_{t}\Big)(v_0+1)b_t.$$
Now, the first step is to show that if  $f_{t+1}^o - \tilde f_{t+1}(\theta)=0$ a.s., then $\omega_0=\omega$, the proof is by contradiction. 
Assume that  $\omega_0\neq\omega$ and  $f_{t+1}^o - \tilde f_{t+1}(\theta)=0$ a.s., then it must be that  $f_{t}^o z_{t}=\omega-\omega_0 \neq 0$ a.s.. Noting that $f_{t}^o$ is independent of $z_{t}$, the only way this is possible is if  both $f_{t}^o$ and $z_{t}$ are constants different from zero. However, the possibility that $f_{t}^o$ has a degenerate distribution is ruled out by $\alpha_{1,0}>0$, therefore $\omega=\omega_0$. As $\omega=\omega_0$ and  $f_{t+1}^o $ is non-zero with probability 1, the only way to have $f_{t+1}^o -\tilde f_{t+1}(\theta)$ a.s. is if $z_t=0$ a.s.. The second step is to show that we need also $\beta=\beta_0$. Using the same argument as before, to have $\beta \neq \beta_0$ and $z_t=0$ a.s.  the random variable $b_t$ has to be constant as $b_t$ is independent of $d_t$. However, $b_t$ is non-constant for any $v_0\in (2,+\infty)$. Therefore, we have that $\beta=\beta_0$. Finally, having   $\beta=\beta_0$, to have $z_t=0$ a.s.  it must be that  $\big(\alpha_{0}-\alpha+(\gamma_{0}-\gamma)d_{t}\big)=0$ a.s.. Indeed, as $d_t$ is non-constant, this is possible only if $\alpha=\alpha_{0}$ and $\gamma=\gamma_{0}$. This concludes the proof.

(C3) This condition is immediately satisfied by Corollary \ref{fpt}.

(C4) From the expression of $l_t(\theta)$ and by an application of the mean value theorem, it results that 
$$|\hat l_t(\theta)- l_t(\theta)|\le |r_t(\theta)| |\hat f_t(\theta) -\tilde f_t(\theta)|,$$
for any $\theta \in \Theta$ and any $t\in \mathbb{N}$. The stochastic coefficient $r_t(\theta)$ has the following expression
$$r_t(\theta)=2^{-1}f^*_t(\theta)^{-1}\left(\frac{(v+1)v^{-1}f^*_t(\theta)y^2_t}{1+v^{-1}f^*_t(\theta)y^2_t}-1\right),$$
 where  $f^*_t(\theta)$ a point between $\tilde f_t(\theta)$ and $\hat f_t(\theta)$. Considering that $\tilde f_t(\theta)$ and $\hat f_t(\theta)$ lie in the set $[c,+\infty )$, $c=\inf_{\theta\in \Theta}\omega/(1-\beta)>0$, it results that 
 \begin{eqnarray*}
\|\hat l_t- l_t\|_\Theta &\le& \| r_t\|_\Theta \|\hat f_t -\tilde f_t\|_\Theta\\
&\le& \bar r  \|\hat f_t -\tilde f_t\|_\Theta,
 \end{eqnarray*}
where
 $$\bar r=2^{-1}c^{-1}\left(1+c^{-1}\left(\max_{\theta \in \Theta}v+1\right)\right).$$
 This shows that C4 is satisfied setting $\eta_t=\bar r$ for any $t\in \mathbb{N}$.

(C5) In view of $\tilde f_0(\theta)\ge \inf_{\theta \in \Theta}\omega/(1-\beta)>0$ a.s. for any $\theta \in \Theta$, it results that
$$\sup_{\theta\in\Theta}l_0(\theta)\le \sup_{\theta \in \Theta}s(v)-\frac{1}{2}\log\left(\inf_{\theta \in \Theta}\omega/(1-\beta)\right)<\infty,$$
 with probability 1. This proves the desired result  $E\|l_0\vee 0\|_\Theta<\infty $.
\end{proof}


\bibliographystyle{apalike}
\bibliography{references}

\begin{thebibliography}{}

\bibitem[Bec et~al., 2008]{bec2008acr}
Bec, F., Rahbek, A., and Shephard, N. (2008).
\newblock {The ACR Model: A Multivariate Dynamic Mixture Autoregression}.
\newblock {\em Oxford Bulletin of Economics and Statistics}, 70(5):583--618.

\bibitem[Berkes et~al., 2003]{berkes2003}
Berkes, I., Horv\'ath, L., and Kokoszka, P. (2003).
\newblock {GARCH processes: structure and estimation}.
\newblock {\em Bernoulli}, 9(2):201--227.

\bibitem[Blasques et~al., 2015]{BSGW2015}
Blasques, F., Gorgi, P., Koopman, S.~J., and Wintenberger, O. (2015).
\newblock {A Note on ``Continuous Invertibility and Stable QML Estimation of
  the EGARCH(1,1) Model''}.
\newblock {\em Tinbergen Institute Discussion Paper 15-131/III}.

\bibitem[Blasques et~al., 2014a]{BSA2014}
Blasques, F., Koopman, S.~J., and Lucas, A. (2014a).
\newblock {Maximum Likelihood Estimation for Generalized Autoregressive Score
  Models}.
\newblock {\em Tinbergen Institute Discussion Paper 14-029/III}.

\bibitem[Blasques et~al., 2014b]{nlgas2014}
Blasques, F., Koopman, S.~J., and Lucas, A. (2014b).
\newblock Optimal formulations for nonlinear autoregressive processes.
\newblock {\em Tinbergen Institute Discussion Paper 14-103/III}.

\bibitem[Blasques et~al., 2014c]{blasques2014se}
Blasques, F., Koopman, S.~J., and Lucas, A. (2014c).
\newblock Stationarity and ergodicity of univariate generalized autoregressive
  score processes.
\newblock {\em Electronic Journal of Statistics}, 8(1):1088--1112.

\bibitem[Bollerslev, 1986]{bol1986}
Bollerslev, T. (1986).
\newblock {Generalized Autoregressive Conditional Heteroskedasticity}.
\newblock {\em Journal of Econometrics}, 31(3):307--327.

\bibitem[Bougerol, 1993]{Bougerol1993}
Bougerol, P. (1993).
\newblock {Kalman Filtering with Random Coefficients and Contractions}.
\newblock {\em SIAM Journal on Control and Optimization}, 31(4):942--959.

\bibitem[Bougerol and Picard, 1992]{bougerol1992}
Bougerol, P. and Picard, N. (1992).
\newblock {Strict Stationarity of Generalized Autoregressive Processes}.
\newblock {\em The Annals of Probability}, 20(4):1714--1730.

\bibitem[Cox, 1981]{cox1981}
Cox, D.~R. (1981).
\newblock {Statistical Analysis of Time Series: Some Recent Developments}.
\newblock {\em Scandinavian Journal of Statistics}, 8(2):93--115.

\bibitem[Creal et~al., 2013]{ckl2013}
Creal, D., Koopman, S.~J., and Lucas, A. (2013).
\newblock {Generalized Autoregressive Score Models with Applications}.
\newblock {\em Journal of Applied Econometrics}, 28(5):777--795.

\bibitem[Davis et~al., 2003]{Davis2003}
Davis, R.~A., Dunsmuir, W. T.~M., and Streett, S.~B. (2003).
\newblock {Observational-driven models for Poisson counts}.
\newblock {\em Biometrika}, 90(4):777--790.

\bibitem[Delle~Monache and Petrella, 2016]{DMP204}
Delle~Monache, D. and Petrella, I. (2016).
\newblock Adaptive models and heavy tails.
\newblock {\em Bank of England Working Paper No. 577}.

\bibitem[Durrett, 2004]{durrett2004}
Durrett, R. (2004).
\newblock {\em Probability: theory and examples}.
\newblock Duxbury Press.

\bibitem[Engle, 1982]{engle1982}
Engle, R.~F. (1982).
\newblock {Autoregressive Conditional Heteroscedasticity with Estimates of the
  Variance of United Kingdom Inflation}.
\newblock {\em Econometrica}, 50:987--1007.

\bibitem[Engle, 2002]{Engle2002}
Engle, R.~F. (2002).
\newblock {Dynamic Conditional Correlation}.
\newblock {\em Journal of Business \& Economic Statistics}, 20(3):339--350.

\bibitem[Engle and Manganelli, 2004]{EM2004}
Engle, R.~F. and Manganelli, S. (2004).
\newblock {Conditional Autoregressive Value at Risk by Regression Quantiles}.
\newblock {\em Journal of Business \& Economic Statistics}, 22(4):367--381.

\bibitem[Engle and Russell, 1998]{EngleRussel1998}
Engle, R.~F. and Russell, J.~R. (1998).
\newblock {Autoregressive Conditional Duration: A New Model for Irregularly
  Spaced Transaction Data}.
\newblock {\em Econometrica}, 66(5):1127--1162.

\bibitem[Francq and Zakoian, 2004]{FZ2004}
Francq, C. and Zakoian, J.~M. (2004).
\newblock {Maximum Likelihood Estimation of Pure GARCH and ARMA-GARCH
  Processes}.
\newblock {\em Bernoulli}, 10(4):605--637.

\bibitem[Francq and Zako{\"\i}an, 2006]{francq2006}
Francq, C. and Zako{\"\i}an, J.-M. (2006).
\newblock {Mixing properties of a general class of GARCH(1,1) models without
  moment assumptions on the observed process}.
\newblock {\em Econometric Theory}, 22(5):815--834.

\bibitem[Glosten et~al., 1993]{GLOSTEN1993}
Glosten, L.~R., Jagannathan, R., and Runkle, D.~E. (1993).
\newblock {On the Relation between the Expected Value and the Volatility of the
  Nominal Excess Return on Stocks}.
\newblock {\em The Journal of Finance}, 48(5):1779--1801.

\bibitem[Granger and Andersen, 1978]{Granger197887}
Granger, C. and Andersen, A. (1978).
\newblock {On the invertibility of time series models}.
\newblock {\em Stochastic Processes and their Applications}, 8(1):87 -- 92.

\bibitem[Harvey, 2013]{H2013}
Harvey, A. (2013).
\newblock {\em Dynamic Models for Volatility and Heavy Tails: With Applications
  to Financial and Economic Time Series}.
\newblock New York: Cambridge University Press.

\bibitem[Harvey and Luati, 2014]{HL2014}
Harvey, A. and Luati, A. (2014).
\newblock {Filtering With Heavy Tails}.
\newblock {\em Journal of the American Statistical Association},
  109(507):1112--1122.

\bibitem[Ito, 2016]{Ryoko2016}
Ito, R. (2016).
\newblock {Asymptotic Theory for Beta-t-GARCH}.
\newblock {\em Cambridge Working Papers in Economics CWPE1607}.

\bibitem[Jensen and Rahbek, 2004]{JRab2004}
Jensen, S.~T. and Rahbek, A. (2004).
\newblock {Asymptotic Inference for Nonstationary GARCH}.
\newblock {\em Econometric Theory}, 20(6):1203--1226.

\bibitem[Lee and Hansen, 1994]{Lee94}
Lee, S. and Hansen, B. (1994).
\newblock {Asymptotic theory for the GARCH(1,1) quasi-maximum likelihood
  estimator}.
\newblock {\em Econometric Theory}, 10(1):29--52.

\bibitem[Lumsdaine, 1996]{Lum1996}
Lumsdaine, R.~L. (1996).
\newblock {Consistency and Asymptotic Normality of the Quasi-Maximum Likelihood
  Estimator in IGARCH(1,1) and Covariance Stationary GARCH(1,1) Models}.
\newblock {\em Econometrica}, 64(3):575--596.

\bibitem[Nelson, 1991]{Nelson1991}
Nelson, D.~B. (1991).
\newblock {Conditional Heteroskedasticity in Asset Returns: A New Approach}.
\newblock {\em Econometrica}, 59(2):347--370.

\bibitem[Newey and West, 1987]{Newey1987}
Newey, W. and West, K. (1987).
\newblock {A Simple, Positive Semi-definite, Heteroskedasticity and
  Autocorrelation Consistent Covariance Matrix}.
\newblock {\em Econometrica}, 55(3):703--08.

\bibitem[Patton, 2006]{Patton2006}
Patton, A.~J. (2006).
\newblock Modelling asymmetric exchange rate dependence.
\newblock {\em International Economic Review}, 47(2):527--556.

\bibitem[Pfanzagl, 1969]{Pfanzagl1969}
Pfanzagl, J. (1969).
\newblock {On the Measurability and Consistency of Minimum Contrast Estimates}.
\newblock {\em Metrika}, 14(1):249--272.

\bibitem[Rao, 1962]{rao1962}
Rao, R.~R. (1962).
\newblock Relations between weak and uniform convergence of measures with
  applications.
\newblock {\em The Annals of Mathematical Statistics}, 33(2):659--680.

\bibitem[Robinson and Zaffaroni, 2006]{robinson2006}
Robinson, P.~M. and Zaffaroni, P. (2006).
\newblock {Pseudo-maximum likelihood estimation of ARCH($\infty$) models}.
\newblock {\em The Annals of Statistics}, 34(3):1049--1074.

\bibitem[Russell, 2001]{Russell2001}
Russell, J.~R. (2001).
\newblock {Econometric Modeling of Multivariate Irregularly-Spaced
  High-Frequency Data}.
\newblock {\em University of Chicago}.

\bibitem[Sorokin, 2011]{sorokin2011}
Sorokin, A. (2011).
\newblock Non-invertibility in some heteroscedastic models.
\newblock {\em Arvix preprint 1104.3318}.

\bibitem[Straumann, 2005]{Straumann2005}
Straumann, D. (2005).
\newblock {Estimation in Conditionally Heteroschedastic Time Series Models}.
\newblock {\em Springer, New York}, 181.

\bibitem[Straumann and Mikosch, 2006]{SM2006}
Straumann, D. and Mikosch, T. (2006).
\newblock {Quasi-Maximum-Likelihood Estimation in Conditionally
  Heteroschedastic Time Series: a Stochastic Recurrence Equation Approach}.
\newblock {\em The Annals of Statistics}, 34(5):2449--2495.

\bibitem[Wald, 1949]{wald1949}
Wald, A. (1949).
\newblock {Note on the Consistency of the Maximum Likelihood Estimate}.
\newblock {\em The Annals of Mathematical Statistics}, 20(4):595--601.

\bibitem[White, 1982]{white1982}
White, H. (1982).
\newblock {Maximum Likelihood Estimation of Misspecified Models}.
\newblock {\em Econometrica}, 50(1):1--25.

\bibitem[Wintenberger, 2013]{Win2013}
Wintenberger, O. (2013).
\newblock {Continuous Invertibility and Stable QML Estimation of the
  EGARCH(1,1) Model}.
\newblock {\em Scandinavian Journal of Statistics}, 40(4):846--867.

\end{thebibliography}

\end{document}